\documentclass[a4paper,11pt]{article}

\usepackage[english]{babel}
\usepackage{amsmath}
\usepackage{amsthm}
\usepackage{enumerate}
\usepackage{graphicx}
\usepackage{mathrsfs}
\usepackage{amssymb}
\usepackage[usenames,dvipsnames]{color}
\usepackage{empheq}
\usepackage{hyperref}
\hypersetup{
colorlinks = true, 
linkcolor = {blue}
}
\usepackage[affil-it]{authblk}

\numberwithin{equation}{section}

\newtheorem{thm}{Theorem}[section]
\newtheorem{prop}[thm]{Proposition}
\newtheorem{lemma}[thm]{Lemma}
\newtheorem{oss}[thm]{Remark}

\def\s{\sigma}

\def\C{{\mathbb C}}

\def\eqref#1{ (\ref{#1})}
\def\&{&\hspace{-20pt}}
\def\R{{\mathbb R}}
\def\d{{\rm d}}

\def\1{\operatorname{Id}}

\def\wt{\widetilde}

\def\bea{\begin{eqnarray}}
\def\eea{\end{eqnarray}}
\def\Id{\mathrm{Id}}
\def\Ai{\mathrm{Ai}}

\def\0{{\bf 0}}

\def\Ai{\mathrm{Ai}}

\def\le{\left}
\def\ri{\right}

\begin{document}
\title{Asymptotics of the Tacnode process: a transition between the gap probabilities from the Tacnode to the Airy process}
\author{Manuela Girotti
\thanks {email: \texttt{mgirotti@mathstat.concordia.ca}}}
\affil{\small Department of Mathematics and Statistics, Concordia University \\
1455 de Maisonneuve Ouest, Montr\'eal, Qu\'ebec, Canada, H3G 1M8}
\date{}

\maketitle

\tableofcontents

\begin{abstract}
We study the gap probabilities of the single-time Tacnode process. Through steepest descent analysis of a suitable Riemann-Hilbert problem, we show that under appropriate scaling regimes the gap probability of the Tacnode process degenerates into a product of two independent gap probabilities of the Airy processes.
\end{abstract}

\section{Introduction}

Determinantal point processes (\cite{Soshnikov}) have played a central role in recent developments of many random models such as random matrix theory, random growth and tiling problems (see for example \cite{DeiftJohansson}, \cite{Johansson2}, \cite{Johansson}).

One of the most well-known model for determinantal point process is the case of $n$ Brownian particles moving on the real line, conditioned never to intersect, with given starting and ending configuration. Let's assume that all the particles start at two given fixed points and end at two other points (which may coincide with the starting points). As time runs in an interval $t \in [0,1]$ ($1$ being the end time where the particles collapse in the two final points), the particles remain confined in a specific region and for every time $t\in [0,1]$ the positions of the Brownian paths form a determinantal process.
Moreover, as the number of particles tends to infinity, such region takes on an explicit shape which depends on the relative position of the starting and ending points and on a parameter $\sigma$ which controls the strength of interaction between the left-most particles and the right-most ones ($\sigma$ can be thought as a pressure or temperature parameter).

There are two possible scenarios: two independent connected components similar to ellipses, one connected component similar to two ``merged" ellipses (see \figurename \ \ref{3cases}, case $(a)$ and $(b)$). 
It is well-known that the microscopic behaviour of such infinite particle system is regulated by the Sine process in the bulk of the particle bundles (\cite{Mehta}), by the Airy process along the soft edges (\cite{Johansson}, \cite{JohanssonAiry}, \cite{TWAiry}) and by the Pearcey process in the cusp singularity (\cite{TWPearcey}), when it occurs. 
\begin{figure}[!h]
\centering
\includegraphics[width=.9\textwidth]{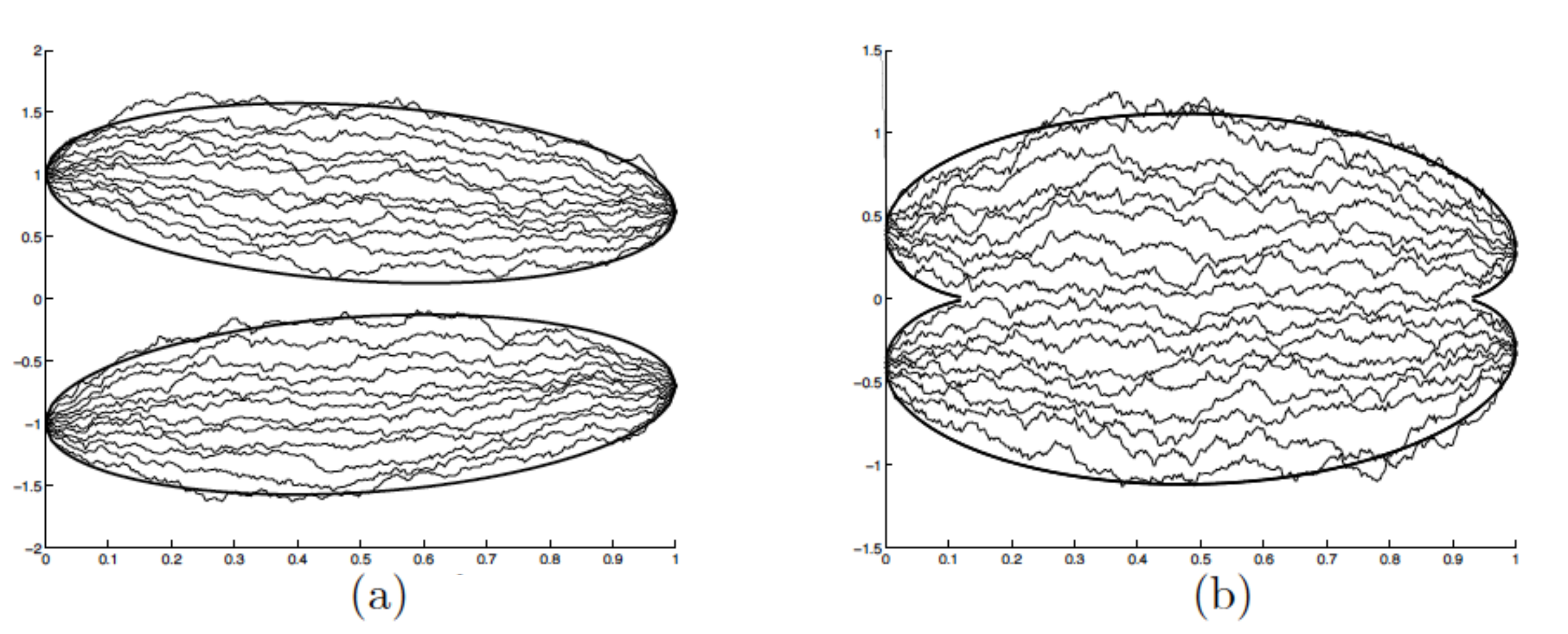}
\begin{center}
\includegraphics[width=.4\textwidth]{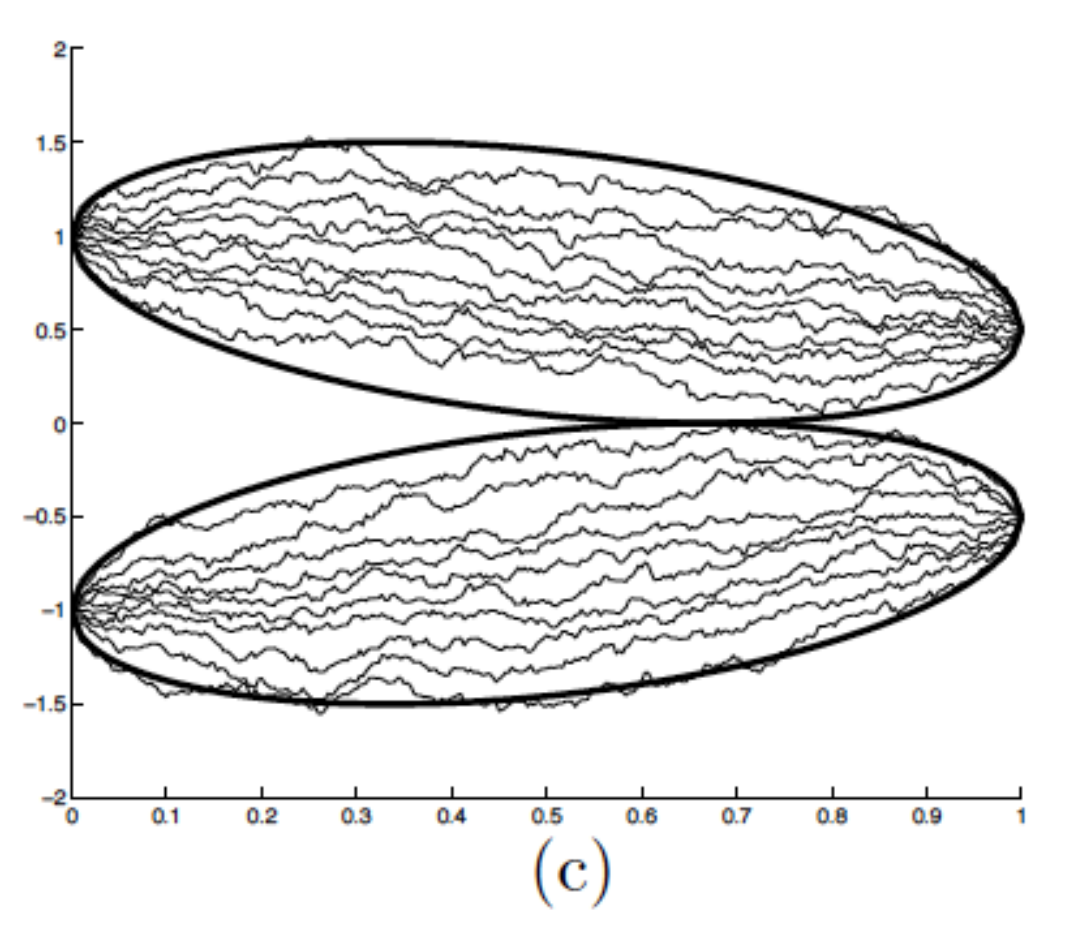}
\end{center}
\caption{Non-intersecting Brownian motions with two starting points $\pm \alpha$ and two ending positions $\pm \beta$ in case of $(a)$ large, $(b)$ small, and $(c)$ critical separation between the endpoints.  
%, and for each fixed $t$ the positions of the $n$ non-intersecting Brownian motions at time $t$ are denoted on the vertical line through $t$. 
For $n\rightarrow \infty$ the positions of the Brownian motions fill a prescribed region in the time-space plane, which is bounded by the boldface lines in the figures. Here the horizontal axis denotes the time, $t \in [0,1]$, the number of paths is $n = 20$ and $(a)$ $\alpha = 1, \beta = 0.7$, $(b)$ $\alpha = 0.4, \beta = 0.3$, and $(c)$ $\alpha = 1, \beta = 0.5$. Taken from \cite{glorystealer}.}
\label{3cases}
\end{figure}

There exist a third critical configuration, which can be seen as a limit of the large separation case when the two bundles are tangential to each other in one point, called tacnode point (see \figurename \ \ref{3cases}, case $(c)$). In a microscopic neighbourhood of this point the fluctuations of the particles are described by a new critical process called Tacnode process.

%The most recently studied case is the one of the so-called tacnode singularity, appearing when the boundary looks (locally) as two circles touching at one point (The limiting mean density for the positions of the Brownian paths at the time of tangency consists of two touching semicircles, possibly of different sizes.). 

The kernel of such process in the single-time case has been first introduced by Kuijlaars \textit{et al.} in \cite{DelvauxKuijZhang}, where the kernel was expressed in terms of a $4\times 4$ matrix valued Riemann-Hilbert problem. Shortly after Kuijlaars's paper, Johansson formulated the multi-time (or extended) version of the process (\cite{JohanssonTac}), remarking nevertheless the fact that this extended version does not automatically reduce to the single-time version given in \cite{DelvauxKuijZhang}. The kernel was expressed in terms of resolvents and Fredholm determinants of the Airy operator acting on a semi-infinite interval $[\sigma, \infty)$. Another version of the multi-time Tacnode process was given in \cite{AdlerFerrariVMoer}. 

In \cite{AdlerJohanssonVMoer} the authors analyzed the same process as arising from random tilings instead of self-avoiding Brownian paths and they proved the equivalency of all the above formulations.
%are indeed equivalent, thus performing a first step in the direction of universality for the tacnode process. 
 A similar result has been obtained by Delvaux in \cite{Delvaux}, where a Riemann-Hilbert expression for the multi-time tacnode kernel is given. A more general formulation of this process has been studied in \cite{FerrariVeto}, where the limit shapes of the two groups of particles are allowed to be non-symmetric.
 
 Physically, if we start from the tacnode configuration and we push together the two ellipses, they will merge giving rise to the single connected component in \figurename \ \ref{3cases}$(b)$, while if we pull the ellipses apart, we simply end up with two disjoint ellipses as in \figurename \ \ref{3cases}$(c)$. It is thus natural to expect that the local dynamic around the tacnode point will in either cases degenerates into a Pearcey process or an Airy process, respectively.
 
 The degeneration Tacnode-Pearcey has been proven in \cite{glorystealer} where the authors showed a uniform convergence of the Tacnode kernel to the Pearcey kernel over compact sets in the limit as the pressure parameter diverges to $-\infty$. On the other hand, the method used in \cite{glorystealer} cannot be extensively applied to the Tacnode-Airy degeneration, since the Airy process can be defined also on non-compact sets. 
 
 The purpose of the present paper is to study the asymptotic behaviour of the gap probability of the (single-time) Tacnode process and its degeneration into the gap probability of the Airy process. 
 There are two types of regimes in which this degeneration occurs:
%  the gap probability for the tacnode process ``degenerates" into gap probabilities of Airy processes: 
  the limit  as $\sigma \rightarrow + \infty$ (large separation), which physically corresponds to pulling apart the two sets of Brownian particles touching on the tacnode point, and the limit as $\tau \rightarrow \pm \infty$ (large time), which corresponds to moving away from the singular point along the boundary of the space-time region swept out by the non-intersecting paths.

An expression for the single-time tacnode kernel is the following (see \cite[formula (19)]{AdlerJohanssonVMoer}) 
\begin{gather}
\mathbb K^{\rm tac} (\tau; x,y) = \nonumber \\ K^{(\tau,-\tau)}_{\Ai} (\s-x,\s-y)
+ \sqrt[3]{2} \int_{\wt\s}^\infty\d z\int_{\wt \s}^\infty \d w\mathcal A^{\tau}_{x-\s}(w)  \left(\1 - K_{\Ai}\bigg|_{[\wt \s, +\infty)}\right)^{-1}(z,w){\mathcal A}^{-\tau}_{y-\s}(z) \label{KtacAJVM}
\end{gather}
with $\wt \s := 2^\frac 2 3 \s$ and 
%the functions appearing in the above definition are specified below:
\begin{eqnarray}
\Ai^{(\tau)}(x)&\& := 
{\rm e}^{\tau x + \frac 2 3 \tau^3} \Ai( x  ) = \int_{\gamma_R}   \frac {\d \lambda}{2i\pi }{\rm e}^{\frac {\lambda^3}3 + \lambda^2 \tau - x \lambda}\\
\Ai(x) &\& := \int_{\gamma_R}   \frac {\d \lambda}{2i\pi }{\rm e}^{\frac {\lambda^3}3 - x \lambda} =  -  \int_{\gamma_L}   \frac {\d \lambda}{2i\pi }{\rm e}^{-\frac {\lambda^3}3 + x \lambda}\\
 \mathcal A^\tau_x(z) &\& : =
 \Ai^{(\tau)}(x + \sqrt[3]{2} z) - \int_0^\infty \d w\, \Ai^{(\tau)} (-x + \sqrt[3]{2} w) \Ai(w+z) \\
K_{\Ai}^{(\tau,-\tau)} (-x, -y) &\& 
%= {\rm e}^{\tau(y- x)} K_{\Ai} (-x+\tau^2, -y+\tau^2)
:= \int_0^\infty\d u  \Ai^{(\tau)}(-x+u)\Ai^{(-\tau)}(-y+u) \\
%&\& =\sqrt[3]{2} \int_0^\infty\!\!\! \d u\,  \Ai^{(\tau)}(-x+\sqrt[3]{2}u)\Ai^{(-\tau)}(-y+\sqrt[3]{2}u)  \\
K_{\Ai}(z, w) &\& 
:= \int_0^\infty\!\!\!\!\d u \, \Ai(z+u)\Ai (w+u) 
\end{eqnarray}
where the contour $\gamma_R$ is the contour extending to infinity in the $\lambda$-plane  along the rays ${e}^{\pm i \frac {\pi}{ 3}}$,  oriented upwards and entirely contained in the right half plane ($\Re( \lambda)>0$), and $\gamma_L := -\gamma_R$.

The quantity of interest, i.e. the gap probability of the process, is expressed in terms of the Fredhom determinant of an integral operator with kernel (\ref{KtacAJVM}). Given a Borel set $\mathcal{I}$, then
\begin{equation}
P\left(\text{no particles in } \mathcal{I}\right) = \det \left( \1 - \mathbb{K}^{\rm tac}\bigg|_{\mathcal{I}}\right)\label{gappropTac}
\end{equation}

The first difficulty in studying the Tacnode process is the expression of its kernel, since it is highly transcendental and it involves the resolvent of the Airy operator. It it thus necessary to reduce it to a more approachable form. 

The first important step was \cite[Theorem 3.1]{MeMnum} where it was proved that gap probabilities of the Tacnode process can be defined as ratio of two Fredholm determinants of explicit integral operators with kernels that only involves contour integrals, exponentials and Airy functions. 
This result, which will be recalled in Section \ref{RHsetting}, will be our starting point in the investigation of the gap probabilities and their asymptotics. 
The second step will be to find an appropriate integral operator in the sense of Its-Izergin-Korepin-Slavnov (\cite{IIKS}) whose Fredholm determinant coincides with the quantity (\ref{gappropTac}). In this way, it will be possible to give a formulation of the gap probabilities of the Tacnode in terms of a Riemann-Hilbert (RH) problem, naturally associated to an IIKS integral operator (see \cite{JohnSasha}). Finally, applying well-known steepest descent methods to the above RH problem, we will be able to prove the conjectured degeneration into Airy processes. 

The RH approach for studying gap probabilities has been extensively used in the past years. To cite a few, we recall the study of gap probabilities for the Airy and Pearcey kernels in \cite{MeMmulti} and \cite{MeM} and for the Bessel kernel in \cite{Me}.

The outline of the paper is the following: in Section \ref{results} we state the main results of the paper, which will be proved in Sections \ref{RHsetting}, \ref{Asymsigma} and \ref{Asymtau}. In particular, Section \ref{RHsetting} deals with some preliminary calculations which are necessary to set a Riemann-Hilbert problem on which we shall later perform some steepest descent analysis in the limit as $\sigma \rightarrow \infty$ (Section \ref{Asymsigma}) or $\tau\rightarrow \infty$ (Section \ref{Asymtau}).

\section{Results}\label{results}
%\paragraph{\textit{Taken from \cite{MeMnum}, plus inspiration from \cite{MeM}}.}
%(what about the general result on multi-interval??)

The first results on asymptotic regime of the tacnode process were stated in \cite{MeMnum}. We are recalling them here for the sake of completeness.
\begin{thm}
Let $\mathcal{I} := \bigcup_{j=1}^K [a_{2j-1}, a_{2j}] $ be collection of intervals, with $a_{j} = a (s_j) = -\sigma -\tau^2 +s_j$. Keeping the overlap $\sigma$ fixed, we have
\begin{equation}
\lim_{\tau \rightarrow \pm \infty} \det \left(\operatorname{Id} - \mathbb K^{\rm tac}\bigg|_{\mathcal{I}}\right) = \det \left(\operatorname{Id}  - K_{\Ai}\bigg|_{J}\right)
\end{equation}
with $J = \bigcup_{\ell=1}^K [s_{2\ell-1}, s_{2\ell}] $. Analogously, keeping $\tau$ fixed, we obtain
\begin{equation}
\lim_{\sigma \rightarrow + \infty} \det \left(\operatorname{Id} - \mathbb K^{\rm tac}\bigg|_{\mathcal{I}}\right) = \det \left(\operatorname{Id}  - K_{\Ai}\bigg|_{J}\right)
\end{equation}
\end{thm}

\begin{proof}
The convergence follows easily by directly studying the kernel of the extended tacnode process (see \cite[formula (19)]{AdlerJohanssonVMoer}), since the term involving the resolvent of the Airy kernel tends to zero, uniformly over
compact sets of the spatial variables $x-\sigma-\tau^2$. 
\end{proof}

A more interesting situation is the one in which the tacnode process degenerates into a couple of Tracy-Widom distributions, in analogy with the Pearcey-to-Airy transition (see \cite{MeM}). In this case, half of the space variables (endpoints of the gaps) moves far away from the tacnode following the left branch of the boundary of the space-time region swept by the particles, and the other half goes in the opposite direction. Therefore, it is expected that the gap probability of the tacnode process for a ``large gap" factorize into two Fredholm determinants for semi-infinite gaps of the Airy process.

Numerically, these regimes are illustrated in \figurename \ \ref{tacnodenum}. The results were already conjectured in \cite{MeMnum} and they are here rigorously proved. 

\begin{figure}
\includegraphics[width=.9\textwidth]{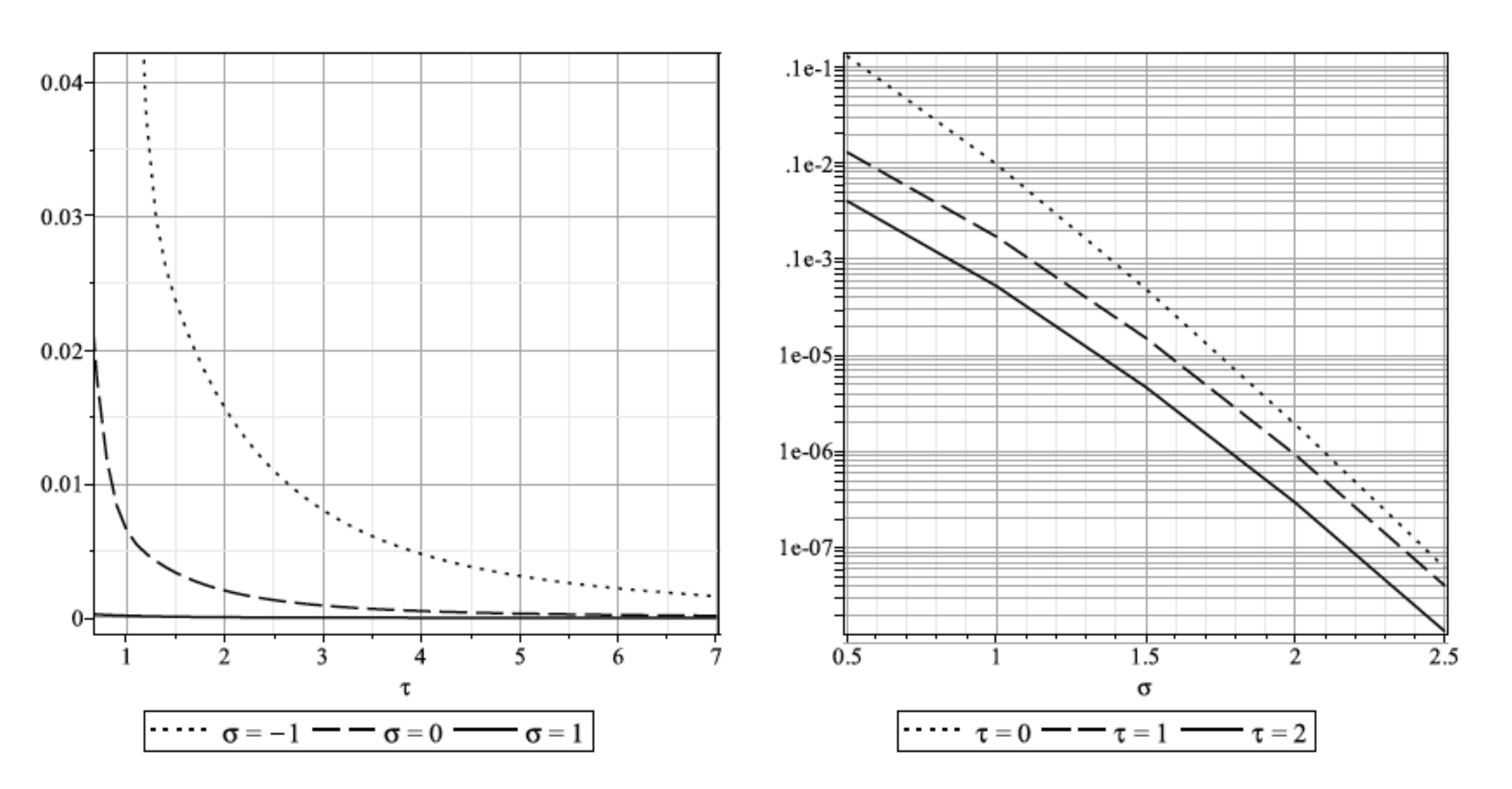}
\caption{The relative values $1- \frac{\det (\1 - \Pi  \mathbb K^{\rm tac} \Pi)}{F_2(a)F_2(b)}$ with $\Pi$ the projection on the interval $[a_{\rm tac}, b_{\rm tac}]$, $a_{\rm tac}= a-\sigma - \tau^2_{tac}$ and $b_{\rm tac}=-b+\sigma +\tau_{tac}^2$, plotted against $\tau_{\rm tac}$ (left) and $\sigma$ (right), showing the convergence of the tacnode gap probability to the product of two Tracy-Widom distributions. Here $a=-0.3$, $b=0.5$. Taken from \cite{MeMnum}. }
\label{tacnodenum}
\end{figure}

In the simple case with only one interval, we have the following theorems.
\begin{thm}[\textbf{Asymptotics as $\sigma \rightarrow +\infty$}]
\label{THMsigma}
Let $\mathbb K^{\rm tac}$ be the tacnode process and $K_{\Ai}$ the Airy process. Let
\begin{equation}
a = a(t) = -\sigma -\tau^2 +t \ \ \ b = b(s) = \sigma + \tau^2 - s
\end{equation}
then as $\sigma \rightarrow +\infty$
\begin{gather}
 \det  \left(\operatorname{Id}  -  \mathbb K^{\rm tac}\bigg|_{[-\sigma -\tau^2 +t, \sigma + \tau^2 -s ]}\right) = \nonumber \\
  \det \left(\operatorname{Id}  - K_{\Ai}\bigg|_{[s,+\infty)}\right) \det \left(\operatorname{Id}  - K_{\Ai}\bigg|_{[t,+\infty)}\right) \left( 1 +  \mathcal{O}(\sigma^{-1}) \right)
\end{gather}
and the convergence is uniform over compact sets of the variables $s,t$ provided
\begin{equation}
-\infty < s,t < K_1 (\sigma + \tau^2), \ \ \ 0<K_1<1.
\end{equation}
\end{thm}

\begin{thm}[\textbf{Asymptotics as $\tau \rightarrow \pm \infty$}]
\label{THMtau}
Let $\mathbb K^{\rm tac}$ be the tacnode process and $K_{\Ai}$ the Airy process. Let
\begin{equation}
a = a(s) = -\sigma -\tau^2 +t \ \ \ b = b(t) = \sigma + \tau^2 - s
\end{equation}
then as $\tau \rightarrow \pm \infty$
 \begin{gather}
 \det   \left(\operatorname{Id}  -  \mathbb K^{\rm tac}\bigg|_{[-\sigma-\tau^2+t, \sigma+\tau^2-s]}\right) = \nonumber \\
  \frac{ \det \left(\operatorname{Id}  - K_{\Ai}\bigg|_{[s,+\infty)}\right) \det \left(\operatorname{Id}  - K_{\Ai}\bigg|_{[t,+\infty)}\right) \det \left(\operatorname{Id}-K_{\Ai}\bigg|_{[\tilde \sigma, \infty)}\right) \left( 1 + \mathcal{O}(\tau^{-1})\right)}{ \det \left(\operatorname{Id}  - K_{\Ai}\bigg|_{[\wt \s,\infty)}\right)}  \nonumber \\
  =    \det \left(\operatorname{Id}  - K_{\Ai}\bigg|_{[s,+\infty)}\right) \det \left(\operatorname{Id}  - K_{\Ai}\bigg|_{[t,+\infty)}\right) \left( 1 + \mathcal{O}(\tau^{-1})\right)
\end{gather}
and the convergence is uniform over compact sets of the variables $s,t$ provided
\begin{gather}
-\infty < s,t < K_1 (\sigma + \tau^2) \\
t = 4\tau^2 - \delta, \ \  0< \delta < \frac{7}{3}K_2\tau^2; \ \ \ \ s = \tau^2 + 2\sigma - \delta, \ \   0<\delta < K_3 \left( 2\sigma + \frac{2}{3}\tau^2\right) 
\end{gather}
for some $0<K_1, \, K_2, \, K_3 <1$.
\end{thm}

More generally, we consider the tacnode process restricted to a collection of intervals.
\begin{thm}
\label{THMgen}
Given
\begin{equation}
\mathcal{I} = \bigcup_{j=1}^J [a_{2j-1}, a_{2j}] \cup [a_{2J+1}, b_0] \cup \bigcup_{k=1}^K[b_{2k-1}, b_{2k}]
\end{equation}
where
\begin{equation}
a_{\ell} = a (s_\ell) = -\sigma -\tau^2 +t_\ell \ \ \ b_\ell = b(t_{2K+1-\ell}) = \sigma +\tau^2 -s_{2K+1-\ell},
\end{equation}
then as $\sigma \rightarrow +\infty$
 \begin{gather}
 \det   \le(\operatorname{Id}  -  \mathbb K^{\rm tac}\bigg|_{\mathcal{I}}\ri) = \det \le(\operatorname{Id}  - K_{\Ai}\bigg|_{J_1}\ri) \det \le(\operatorname{Id}  - K_{\Ai}\bigg|_{J_2}\ri) \left( 1 + \mathcal{O}(\sigma^{-1})\right)
 \end{gather}
  or as $\tau \rightarrow \pm \infty$
  \begin{gather}
 \det   \le(\operatorname{Id}  -  \mathbb K^{\rm tac}\bigg|_{\mathcal{I}}\ri) = \det \le(\operatorname{Id}  - K_{\Ai}\bigg|_{J_1}\ri) \det \le(\operatorname{Id}  - K_{\Ai}\bigg|_{J_2}\ri) \left( 1 + \mathcal{O}(\tau^{-1})\right)
\end{gather}
where
\begin{equation}
J_1 = \bigcup_{\ell=1}^J [t_{2\ell-1}, t_{2\ell}] \cup [t_{2J+1}, +\infty) \ \ \ J_2=\bigcup_{\ell=1}^K[s_{2\ell-1}, s_{2\ell}] \cup [s_{2K+1}, +\infty)
\end{equation}
and the convergence is uniform over compact sets of the variables $s,t$ provided
\begin{gather}
-\infty < s_\ell,t_\ell < K_1 (\sigma + \tau^2) \\
t_\ell = 4\tau^2 - \delta, \ \  0< \delta < \frac{7}{3}K_2\tau^2; \ \ \ \ s_\ell = \tau^2 + 2\sigma - \delta, \ \   0<\delta < K_3 \left( 2\sigma + \frac{2}{3}\tau^2\right) 
\end{gather}
for some $0<K_1, \, K_2, \, K_3 <1$.
\end{thm}

The parametrization of the endpoints $a$ and $b$ in Theorems \ref{THMsigma} and \ref{THMtau} (and of $a_\ell$ and $b_\ell$ in Theorem \ref{THMgen}) has the following meaning. 
%The tacnode process arises in the study of self-avoiding random walks on the line, with two starting points (at time $t=0$) and two end points (at time $t=T$).
%%conditioned to start at two distinct points $\kappa_1$ and $\kappa_2$ at time $0$ and end at time $T>0$ at the points.
%At any time $0<t<T$ the bulk of the walkers consists of either one or two finite intervals, depending on the position of the source-sink points and on the ``pressure" parameter $\sigma$, which may push together the two bulks so that they merge or can pull them apart so that they are disjoint. 
%For specific values of $\sigma$ and positions of the points we have a new setting: 
At the critical time $0< t_{\rm tac} <1$, the two bulks tangentially touch at one point $P_{\rm tac}$, the tacnode point. From the common tacnode point $a(t_{\rm tac})=b(t_{\rm tac})$, two new endpoints $[a(t), b(t)]$ emerge and move away along the branches of the boundary.
% according to a(t)=ac + O(t ? tc)
%32

The tacnode point process describes the statistics of the random walkers in a scaling neighborhood of $t=t_{\rm tac}$ and $a=b=P_{\rm tac}$.
%more precisely, we rescale t 
%? tc + N?1/2t, a
%? ac + N?1/4a, and b 
%? ac + N?1/4b; see, for
%instance [33]. 
The asymptotics as $\tau \rightarrow \pm \infty$ given in Theorem \ref{THMtau} is the regime where we look ``away" from the critical point (either in the future for $\tau >0$ or in the past for $\tau <0$) and it is expected to reduce to two Airy point processes, which describe the edge-behavior of the random walkers. Similarly, when we take the limit as $\sigma \rightarrow +\infty$ (Theorem \ref{THMsigma}) we are physically pushing away the two bulks from each other and the expected regime around the not-any-more critical time will be again a product of two Airy point processes.

The proof of these theorems rely essentially upon the construction of a Riemann-Hilbert problem deduced from a suitable IIKS integrable kernel (\cite{IIKS}) and the Deift-Zhou steepest descent method (\cite{SDM}). In the next section we will show how to deduce such integrable kernel from the Tacnode kernel.  We will start with considerations that apply to the more general case, but then we will specialize to the single interval case (Theorems \ref{THMsigma} and \ref{THMtau}) in order to avoid unnecessary complications, which are purely notational and not conceptual.

\section{The Riemann-Hilbert setting for the gap probabilities of the Tacnode process}\label{RHsetting}

We recall the definition of the tacnode kernel, referring to the formula given by Adler, Johansson and Van Moerbeke in \cite{AdlerJohanssonVMoer}.

The single-time tacnode kernel reads (see \cite[formula (19)]{AdlerJohanssonVMoer}) 
\begin{gather}
\mathbb K^{\rm tac} (\tau; x,y) = \nonumber \\ K^{(\tau,-\tau)}_{\Ai} (\s-x,\s-y)
+ \sqrt[3]{2} \int_{\wt\s}^\infty\d z\int_{\wt \s}^\infty \d w\mathcal A^{\tau}_{x-\s}(w)  \left(\1 - K_{\Ai}\bigg|_{[\wt \s, +\infty)}\right)^{-1}(z,w){\mathcal A}^{-\tau}_{y-\s}(z)
\end{gather}
where $\wt \s := 2^\frac 2 3 \s$ and the functions appearing in the above definition are specified below:
\begin{eqnarray}
\Ai^{(\tau)}(x)&\& := 
{\rm e}^{\tau x + \frac 2 3 \tau^3} \Ai( x  ) = \int_{\gamma_R}   \frac {\d \lambda}{2i\pi }{\rm e}^{\frac {\lambda^3}3 + \lambda^2 \tau - x \lambda}\\
\Ai(x) &\& := \int_{\gamma_R}   \frac {\d \lambda}{2i\pi }{\rm e}^{\frac {\lambda^3}3 - x \lambda} =  -  \int_{\gamma_L}   \frac {\d \lambda}{2i\pi }{\rm e}^{-\frac {\lambda^3}3 + x \lambda}\\
 \mathcal A^\tau_x(z) &\& : =
 \Ai^{(\tau)}(x + \sqrt[3]{2} z) - \int_0^\infty \d w\, \Ai^{(\tau)} (-x + \sqrt[3]{2} w) \Ai(w+z) \\
K_{\Ai}^{(\tau,-\tau)} (-x, -y) &\& 
%= {\rm e}^{\tau(y- x)} K_{\Ai} (-x+\tau^2, -y+\tau^2)
:= \int_0^\infty\d u  \Ai^{(\tau)}(-x+u)\Ai^{(-\tau)}(-y+u) \\
%&\& =\sqrt[3]{2} \int_0^\infty\!\!\! \d u\,  \Ai^{(\tau)}(-x+\sqrt[3]{2}u)\Ai^{(-\tau)}(-y+\sqrt[3]{2}u)  \\
K_{\Ai}(z, w) &\& 
:= \int_0^\infty\!\!\!\!\d u \, \Ai(z+u)\Ai (w+u) 
\end{eqnarray}
The contour $\gamma_R$ is a contour extending to infinity in the $\lambda$-plane  along the rays ${e}^{\pm i \frac {\pi}{ 3}}$,  oriented upwards and entirely contained in the right half plane ($\Re( \lambda)>0$), and $\gamma_L := -\gamma_R$.

%Our goal is to present an isomonodromic system (Riemann--Hilbert problem) for the gap probability (single-time)
%\begin{equation}
%F_E:= \det \le( \1 - \mathbb K^{\rm tac}\bigg|_E\ri)
%\end{equation}
%where $E$ is an arbitrary finite  union of bounded disjoint intervals.
%Since only the combination $x-\sigma, y-\sigma$ appears, we can shorten formul\ae by a corresponding shift, just remembering that the projectors to be applied to $x, y$ need to be counter-shifted.
%The shifted kernel thus reads (we re-use the same symbol)
First of all, since only the combination $x-\sigma, y-\sigma$ appears, we shift the variables and we perform a spacial rescaling of the form $u=\sqrt[3]{2} u'$. The resulting kernel is
%\begin{gather}
% \mathbb K^{\rm tac} (\tau; x,y):= \mathbb K^{\rm tac} (\tau; x+\s,y+\s) = \nonumber \\
% = K^{(\tau,-\tau)}_{\Ai} (-x,-y) + \sqrt[3]{2} \int_{\wt\s}^\infty\d z\int_{\wt \s}^\infty \d w\mathcal A^{\tau}_{x}(w)  (\1 - K_{\Ai})_{\wt \s}^{-1}(z,w){\mathcal A}^{-\tau}_{y}(z)
% \end{gather}
%For convenience in the formul\ae \ below we shall perform a spatial rescaling $x:=\sqrt[3]{2} x', \ y:= \sqrt[3]{2} y'$ 
% \begin{equation}
%  \wt {\mathbb K}(x',y') =  {\sqrt[3]{2}} \mathbb K (\sqrt[3]{2} x', \sqrt[3]{2} y')
%  \end{equation}
%  Renaming $x', y'$ by $x,y$ we have 
\begin{gather}
\wt {\mathbb K} (x,y) :=  {\sqrt[3]{2}} \mathbb K^{\rm tac}  (\sqrt[3]{2} x, \sqrt[3]{2} y) 
= {\sqrt[3]{2}} \int_0^\infty \d u \Ai^{(\tau)}  (\sqrt[3]{2} (u -x )) \Ai^{(-\tau)} (\sqrt[3]{2} (u-y)) +
\nonumber \\+
 {\sqrt[3]{2}}  \int_{\wt \s}^\infty \!\!\! \d z \int_{\wt \s }^{\infty} \!\!\! \d w
\mathcal A_{\sqrt[3]{2} x}^{\tau} (w)\le(\1 - K_{\Ai} \bigg|_{\wt \s}\ri)^{-1}(z,w)  \mathcal A_{\sqrt[3]{2} y}^{-\tau} (w) 
\end{gather}

For the sake of brevity, we shall introduce the operators $K_\Ai$, $K^{(\tau,-\tau)}_\Ai$, $\mathfrak A_\tau$ (with abuse of notation) as the operators with the kernels,
\begin{eqnarray}
&&K^{(\tau,-\tau)}_\Ai := K^{(\tau,-\tau)}_\Ai (\sqrt[3]{2} x , \sqrt[3]{2}y)=   \sqrt[3]{2} \int_0^\infty \d u \Ai^{(\tau)}  (\sqrt[3]{2} (u -x )) \Ai^{(-\tau)} (\sqrt[3]{2} (u-y))  \\
&&K_\Ai :=  K_\Ai(x,y)\bigg|_{[\wt \s,\infty)}\\
&& \mathcal B_\tau  (x,z):= 2^{\frac{1}{6}} \Ai^{(\tau)} \le(\sqrt[3]{2} (x+ z)\ri), \ \ \  \mathcal A(z,w):= \Ai(z+w) \\
&& \mathfrak A_\tau (x,z) := \mathcal A^\tau_{\sqrt[3]{2} x} (z) = \mathcal B_\tau(x,z) - \int_0^\infty\!\!\!\! \d w \,\mathcal B_\tau(-x,w) \mathcal A(w, z)\\
%&& \pi = \chi_{z\geq \wt \s}: L^2(\R)\to L^2(\R) \ \ \ \hbox { (projection operator) } \\
%&& \Pi  = \chi_{x\in E}: L^2(\R)\to L^2(\R) \ \ \ \hbox { (projection onto $E$)}
\end{eqnarray}
moreover, we set $\pi$ as the projector on the interval $[\wt \s, \infty)$.

Given the above definitions, we can rewrite the tacnode kernel in the following way
\begin{prop}
The kernel $\wt {\mathbb K}$ can be represented as 
    \begin{eqnarray}
  \wt {\mathbb K}(x,y) &\& = K_\Ai^{(\tau,-\tau)}(x,y) + \int_{[\wt \s, \infty)} \!\!\! \d z  \int_{[\wt \s, \infty)}\!\!\! \d w\, \mathfrak A_\tau(x,z) \mathcal R(z,w) \mathfrak A_{-\tau} (z,y)   
   \\
%   &\&=   \int_{\R_+} \!\!\!\! \d u \,  \mathcal B_\tau(-x,u)\mathcal B_{-\tau}(-y,u) + \\
%  &\&+ \int\!\!\!\int_{z,w\geq \wt \s}\!\!\!\!\!\!\!\! \d z \d w\le[ \mathcal B_\tau(x,z) - \int_{\R_+} \!\!\!\! \d u \, \mathcal B_\tau(-x,u) \mathcal{A}(u,z)\ri] \mathcal R(z,w)
%  \le[ \mathcal B_{-\tau}(y,w) -\int_{\R_+} \!\!\!\! \d u \,  \mathcal B_{-\tau}(-y,u) \mathcal A (u,w)\ri] \cr
  \mathcal R(z,w) &\& := \le(\1 - K_{\Ai}\bigg|_{[\wt \s,\infty)}\ri)^{-1}(z,w) 
  \end{eqnarray}
Alternatively, 
\begin{equation}
\wt {\mathbb K} = K_\Ai^{(\tau,-\tau)}+ \mathfrak A_\tau \pi (\1- K_\Ai)^{-1} \pi \mathfrak A_{-\tau}^T 
\end{equation}
where we recall that $\wt {\mathbb K}$ is the transformed of the kernel $\mathbb K^{\rm tac} $ under the change of variables $u' = 2^{-\frac{1}{3}}(u-\sigma)$.
  \end{prop}
  
Let $\mathcal{I} = [a_1, a_2]\sqcup [a_3,a_4] \dots \sqcup [a_{2K-1}, a_{2K}]$ and denote by  $\Pi$ the projector on $\mathcal{I}$. We will denote with $\tilde \Pi$ the projection on the rescaled and translated collection of intervals $[\tilde a_1, \tilde a_2] \sqcup \ldots \sqcup [\tilde a_{2K-1}, \tilde a_{2K}]$, where $\tilde a_j := 2^{-\frac{1}{3}}(a_j-\sigma)$. We are interested in studying the gap probability of the Tacnode process restricted to this collection of intervals, namely  
\begin{equation}
\det( \1 - \Pi \mathbb K^{\rm tac} \Pi) = \det \left( \1 - 2^{\frac{1}{3}} \tilde \Pi \le( K_\Ai^{(\tau,-\tau)} + \mathfrak A_\tau\pi \left(\1 - K_\Ai \bigg|_{[\wt \s, \infty)}\right)^{-1} \pi \mathfrak A_{-\tau}^T\ri)\tilde \Pi\right)
\end{equation}

%Let $\mathcal{I} = [a_1, a_2]\sqcup [a_3,a_4] \dots \sqcup [a_{2K-1}, a_{2K}]$ and denote by  $\Pi$ the projector on $\mathcal{I}$. We are interested in studying the gap probability of the Tacnode process restricted to this collection of intervals, namely  
%\begin{equation}
%\det( \1 - \Pi \wt {\mathbb K} \Pi) = \det \left( \1 - \Pi \le( K_\Ai^{(\tau,-\tau)} + \mathfrak A_\tau\pi \left(\1 - K_\Ai \bigg|_{[\wt \s, \infty)}\right)^{-1} \pi \mathfrak A_{-\tau}^T\ri)\Pi\right)
%\end{equation}

The following proposition is a restatement of Theorem 3.1 from \cite{MeMnum}, adapted to the single-time case which we are examining.
\begin{prop} \label{conditionedprob}
%\marginpar{Note: the minus in the (1,2) block may also be replaced by a plus.}
The gap probability of the Tacnode process admits the following equivalent representation
\begin{gather}
\det( \1 - \Pi \mathbb K^{\rm tac}\Pi)  = F_2(\wt \s)^{-1} \det \le(\1 - \hat \Pi \mathbb{H} \hat \Pi  \right) := \nonumber \\
=F_2(\wt \s)^{-1} \det \le(\1  - \sqrt[3]{2} \le[ \begin{array}{c|c}
 \pi K_\Ai \pi & -\sqrt[6]{2} \pi \mathfrak A_{-\tau}^T\tilde \Pi\\
 \hline
 - \sqrt[6]{2} \tilde \Pi \mathfrak A_{\tau} \pi & \sqrt[3]{2} \tilde \Pi K_\Ai^{(\tau,-\tau)} \tilde \Pi 
\end{array}\ri]\ri)   \label{explode0}
\end{gather}
%\begin{equation}
%\det( \1 - \Pi \wt {\mathbb K} \Pi)  = F_2(\wt \s)^{-1} \det \le(\1 - \hat \Pi \mathbb{H} \hat \Pi  \right) =F_2(\wt \s)^{-1} \det \le(\1  - \le[ \begin{array}{c|c}
% \pi K_\Ai \pi & -\pi \mathfrak A_{-\tau}^T\Pi\\
% \hline
% -\Pi \mathfrak A_{\tau} \pi & \Pi K_\Ai^{(\tau,-\tau)}\Pi 
%\end{array}\ri]\ri)   \label{explode0}
%\end{equation}
where  $\hat \Pi \mathbb{H} \hat \Pi$ is an operator acting on the Hilbert space $L^2([\tilde \sigma, \infty)) \oplus L^2(\R)$,  $\hat \Pi := \pi \oplus \tilde \Pi$ and $F_2(\tilde \sigma)$ is the Tracy--Widom distribution 
\begin{equation}
F_2(\wt \s):= \det \le(\1 - K_{\Ai}\bigg|_{[\wt \s,\infty)}\ri).
\end{equation}
%and the operator appearing on the right hand side of formula (\ref{explode0}) acts on $L^2([\tilde \sigma,\infty))\oplus L^2(\mathbb{R})$.
\end{prop}
\begin{oss}
The projection $\pi$ in (\ref{explode0}) is redundant since by definition the operator acts on the Hilbert space $L^2([\wt \sigma,\infty))$, but we will keep it for convenience.
\end{oss}

The gap probabilities of the Tacnode process are expressible as ratio of two Fredholm determinants. Therefore, we can interpret the tacnode process as a (formal) conditioned process: its gap probabilities are the gap probabilities of the process $\mathbb{H}$ conditioned such that there are no points in the interval $[\wt \s, \infty)$.

\begin{proof}
The identity is based on the following operator identity (all being trace-class perturbations of the identity)
\begin{gather}
 \det \le(\1  - \le[ \begin{array}{c|c}
 \pi K_\Ai\pi & -\sqrt[6]{2} \pi \mathfrak A_{-\tau}^T\tilde \Pi\\
 \hline
 -\sqrt[6]{2} \tilde \Pi \mathfrak A_{\tau} \pi & \sqrt[3]{2} \tilde \Pi K_\Ai^{(\tau,-\tau)} \tilde \Pi 
\end{array}\ri]\ri) =
 \det\le[\begin{array}{c|c}
 \1 -\pi K_\Ai \pi  & 0 \\ 
 \hline 
 0&\1 \end{array}
 \ri] 
 \det  \le[
\begin{array}{c|c}
\1 & 0\\
\hline
\sqrt[6]{2} \tilde \Pi \mathfrak A_\tau \pi & \1
\end{array}
\ri] \times \nonumber \\
\times \det \le[ \begin{array}{c|c}
 \1 & \sqrt[6]{2} (\1-K_\Ai)^{-1}_{\wt \s} \pi \mathfrak A_{-\tau}^T \tilde \Pi\\
 \hline
 0 & \1 -  \sqrt[3]{2} \left\{ \tilde \Pi K_\Ai^{(\tau,-\tau)}\tilde \Pi - \tilde \Pi \mathfrak A_\tau \pi (\1-K_\Ai)^{-1}_{\wt \s}\pi \mathfrak A_{-\tau}^T \tilde \Pi \right\}
\end{array}\ri]  \cr
=\det \le(\1- \pi K_\Ai \pi  \ri) \det( \1 - \tilde \Pi \wt{\mathbb K} \tilde \Pi) 
\end{gather}
%\begin{gather}
% \det \le(\1  - \le[ \begin{array}{c|c}
% \pi K_\Ai\pi & -\pi \mathfrak A_{-\tau}^T\Pi\\
% \hline
% -\Pi \mathfrak A_{\tau} \pi & \Pi K_\Ai^{(\tau,-\tau)}\Pi 
%\end{array}\ri]\ri) =\cr
% \det\le[\begin{array}{c|c}
% \1 -\pi K_\Ai \pi  & 0 \\ 
% \hline 
% 0&\1 \end{array}
% \ri] \det  \le[
%\begin{array}{c|c}
%\1 & 0\\
%\hline
%\Pi \mathfrak A_\tau \pi & \1
%\end{array}
%\ri]
%\det \le[ \begin{array}{c|c}
% \1 & (\1-K_\Ai)^{-1}_{\wt \s} \pi \mathfrak A_{-\tau}^T\Pi\\
% \hline
% 0 & \1 -  \Pi K_\Ai^{(\tau,-\tau)}\Pi - \Pi \mathfrak A_\tau \pi (\1-K_\Ai)^{-1}_{\wt \s}\pi \mathfrak A_{-\tau}^T\Pi
%\end{array}\ri]  \cr
%=\det \le(\1- \pi K_\Ai \pi  \ri) \det( \1 - \Pi \wt{\mathbb K} \Pi) 
%\end{gather}
\end{proof}

Our next goal is to find suitable Fourier representations of the various operators appearing in (\ref{explode0}). In order to do that, we will rewrite the kernels involved, with their projections respectively, in terms of contour integrals. The results are shown in the following two lemmas. Their proof is just a matter of straightforward calculations using Cauchy's residue theorem.  
  \begin{lemma}
The kernels involved in the definitions can be represented as the following contour integrals
\begin{eqnarray}
\mathcal B^\tau(x,z) &\&  = 2^{-\frac{1}{6}} \int_{\gamma_R} \frac {\d \lambda}{2i\pi} {\rm e}^{\wt\theta_{\tau}(\lambda;x+z)},  \ \ \ 
%&\& =\frac 1{\sqrt[3]{2}} \int_{\gamma_R} \frac {\d \lambda}{2i\pi} 
%{\rm e}^{\lambda^3/6 + 2^{-\frac 23} \tau \lambda^2 -  (x+z)\lambda} =\\
\mathcal A(z,w) = \int_{\gamma_R} \frac {\d \lambda}{2i\pi} {\rm e}^{\theta(\lambda;x+z)} \\
 \mathfrak A_\tau (x,z)&\& = 2^{-\frac{1}{6}} \le[ - 
\int_{\gamma_L} \frac {\d\mu}{2\pi i} {\rm e}^{ -\wt \theta_{-\tau}(\mu; x+z)} 
-\int_{\gamma_L} \frac {\d\mu}{2\pi i}\int_{\gamma_R} \frac {\d\lambda}{2\pi i}
 \frac { {\rm e}^{ - \wt \theta_{-\tau}(\mu; -x) + \theta(\lambda;z)}}{\mu-\lambda}
\ri] \\
 K_\Ai^{(\tau,-\tau)}(x,y) &\&= \int_{\gamma_R}\! \frac{\d\lambda}{2i\pi}  \int_{\gamma_L}\!
\frac {\d\mu}{2i\pi} \frac { {\rm e}^{-\wt \theta_{-\tau}(\mu, -x) + \wt \theta_{-\tau}(\lambda, -y) } }{\sqrt[3]{2} (\mu - \lambda )}  \\
K_\Ai(z,w) &\& :=   \int_{\gamma_R}  \frac{\d\lambda}{2i\pi}      \int_{\gamma_L}\frac {\d\mu}{2i\pi}     \frac { {\rm e}^{ \theta(\lambda, z) - \theta(\mu, w) } } {\mu - \lambda } 
\end{eqnarray}
with $\wt \theta_\tau(\lambda; x) := \frac {\lambda^3}6  + \frac {\tau}{2^{2/3}} \lambda^2 -x\lambda$ and $\theta(\lambda;x):= \frac {\lambda^3}3 - x\lambda$.
\end{lemma}

Moreover, giving the projector $\tilde \Pi$ on the collection of intervals $\bigcup_j [\tilde a_{2j-1}, \tilde a_{2j}]$, the following identities hold
\begin{eqnarray}
&\& \tilde \Pi \mathfrak A_\tau \pi (x,w) =\\
&\&=\sum_{j} ^{2K} \frac {(-1)^j}  {\sqrt[3]{2}}
\int_{i\R} \frac{\d \xi}{2i\pi} {\rm e}^{\xi( \tilde a_j - x)}
\int_{i\R} \frac {\d \zeta}{2i\pi} {\rm e}^{ \zeta(\wt \s - w)}
  \le[
\int_{\gamma_R} \frac {\d\lambda}{2\pi i} \frac{{\rm e}^{ \wt \theta_\tau(\lambda; \tilde a_j+\wt \s)}}
{(\xi - \lambda)(\lambda-\zeta)} 
-\int_{\gamma_L} \!\!\!\frac {\d\mu}{2\pi i}\int_{\gamma_R} \!\!\!\frac {\d\lambda}{2\pi i}
 \frac { {\rm e}^{ - \wt \theta_{-\tau}(\mu; -\tilde a_j) + \theta(\lambda;\wt \s)}}
 {(\mu-\lambda)(\xi - \mu)(\lambda - \zeta)}
\ri]\nonumber
\end{eqnarray}
\begin{eqnarray}
&\& \pi \mathfrak A_{-\tau}^T \tilde \Pi (z,y) =\\
&\& = \sum_{j} \frac {(-1)^{j}}  {\sqrt[3]{2}} \int_{i\R} \frac {\d\xi}{2i\pi} {\rm e}^{\xi(z - \wt \s )}
 \int_{i\R} \frac {\d\zeta}{2i\pi} {\rm e}^{\zeta(y - \tilde a_j )}
\le[
-\int_{\gamma_L}\!\!\! \frac {\d\mu}{2\pi i} \frac{{\rm e}^{- \wt \theta_{\tau}(\mu; \tilde a_j+\wt \s)}}
{(\xi - \mu)( \mu - \zeta )} 
-\int_{\gamma_L} \!\!\! \frac {\d\mu}{2\pi i}\int_{\gamma_R}\!\!\! \frac {\d\lambda}{2\pi i}
 \frac { {\rm e}^{  \wt \theta_{-\tau}(\lambda; -\tilde a_j) - \theta(\mu;\wt \s)}}
 {(\mu-\lambda)(\xi - \mu)(\lambda - \zeta)}
\ri]\nonumber
\end{eqnarray}
\begin{eqnarray}
&\& \tilde \Pi K_\Ai^{(\tau,-\tau)}(x,y) \tilde \Pi = \\
&\& = \sum_{j,k} \frac{(-1)^{j+k}}{\sqrt[3]{4}} 
\int_{i\R}\frac {\d\xi}{2i\pi} {\rm e}^{\xi(\tilde a_j-x)}
\int_{i\R}\frac {\d\zeta}{2i\pi} {\rm e}^{\zeta(y - \tilde a_k)}
 \int_{\gamma_R}\frac{\d\lambda}{2i\pi}  \int_{\gamma_L}\frac{\d\mu}{2i\pi} \frac { {\rm e}^{ - \wt \theta_{-\tau}(\mu, -\tilde a_j) + \wt \theta_{-\tau}(\lambda, -\tilde a_k) } }
 {(\mu - \lambda )(\xi -\mu)(\lambda - \zeta)} \nonumber \\
&\&\pi K_\Ai \pi(z,w)  =\int_{i\R}\frac{\d \xi}{2i\pi} {\rm e}^{\xi(z - \wt \s)} 
\int_{i\R}\frac{\d \zeta}{2i\pi} {\rm e}^{\zeta( \wt \s- w)}  \int_{\gamma_R}\frac{\d\lambda}{2i\pi}  \int_{\gamma_L} 
\frac{\d\mu}{2i\pi} \frac { {\rm e}^{ \theta(\lambda, \wt \s) - \theta(\mu, \wt\s) } } {(\mu - \lambda )(\xi - \mu)(\lambda-\zeta)  }
\end{eqnarray}

%We can now perform the Fourier transform of the above operators.
%: we will use the convention that the sign of the transform in $x,y$ or $w,z$ is the opposite.
%\begin{equation}
%f(x) = \int_{i\R} {\rm e}^{-\xi x} \wh f(\xi) \frac{ {\rm d} {\xi}}{2i\pi}\ \ \ \ \  
%h(z) =\int_{i\R} {\rm e}^{z\xi} \wh h(\xi) \frac{ {\rm d} {\xi}}{2i\pi} 
%\end{equation}
\begin{lemma}
\label{lemmaFourier}
The Fourier representation of the previous operators is the following
\begin{eqnarray}
&\& \mathcal F(\tilde \Pi \mathfrak A_\tau\pi)(\xi, \zeta)=\\
&\&=
\sum_{j} \frac {(-1)^j}  {2i\pi \sqrt[3]{2}} {\rm e}^{\xi \tilde a_j+ \zeta \wt \s}
  \le[
\int_{\gamma_R} \frac {\d\lambda}{2\pi i} \frac{{\rm e}^{ \wt \theta_\tau(\lambda; \tilde a_j+\wt \s)}}
{(\xi - \lambda)(\lambda-\zeta)} 
-\int_{\gamma_L} \!\!\!\frac {\d\mu}{2\pi i}\int_{\gamma_R} \!\!\!\frac {\d\lambda}{2\pi i}
 \frac { {\rm e}^{ - \wt \theta_{-\tau}(\mu; -\tilde a_j) + \theta(\lambda;\wt \s)}}
 {(\mu-\lambda)(\xi - \mu)(\lambda - \zeta)}
\ri]\nonumber
\end{eqnarray}
\begin{eqnarray}
&\& \mathcal F(\pi  \mathfrak A_{-\tau}^T\tilde \Pi  )(\xi, \zeta)=\\
&\& = \sum_{k} \frac {(-1)^k}  {2i\pi \sqrt[3]{2}} {\rm e}^{ - \wt \s \xi - \tilde a_k \zeta}
\le[
-\int_{\gamma_L} \frac {\d\mu}{2\pi i} \frac{{\rm e}^{- \wt \theta_{\tau}(\mu; \tilde a_k+\wt \s)}}{(\xi-\mu)( \mu - \zeta )} 
-\int_{\gamma_L} \frac {\d\mu}{2\pi i}\int_{\gamma_R} \frac {\d\lambda}{2\pi i}
 \frac { {\rm e}^{  \wt \theta_{-\tau}(\lambda; -\tilde a_k) - \theta(\mu;\wt \s)}}{(\mu-\lambda)(\xi - \mu)(\lambda - \zeta)}
\ri] \nonumber
\end{eqnarray}
\begin{eqnarray}
&\& \mathcal F( \tilde \Pi K_\Ai^{(\tau,-\tau)}\tilde \Pi) (\xi, \zeta)=
 \sum_{j,k} \frac{(-1)^{j+k}}{2i\pi \sqrt[3]{4}} 
{\rm e}^{  \tilde a_j \xi  - \tilde a_k \zeta}
 \int_{\gamma_R}\frac{\d\lambda}{2i\pi}  \int_{\gamma_L}\frac{\d\mu}{2i\pi} \frac { {\rm e}^{ - \wt \theta_{-\tau}(\mu, -\tilde a_j) + \wt \theta_{-\tau}(\lambda, -\tilde a_k) } }{(\mu - \lambda )( \xi-\mu)( \lambda - \zeta)} \\
&\& \mathcal F( \pi K_\Ai\pi)(\xi, \zeta)  =\frac1{2i\pi}  {\rm e}^{ \wt \s (\zeta -  \xi)  }
  \int_{\gamma_R}\frac{\d\lambda}{2i\pi}  \int_{\gamma_L}
\frac{\d\mu}{2i\pi} \frac { {\rm e}^{ \theta(\lambda, \wt \s) - \theta(\mu, \wt\s) } } {(\mu - \lambda )(\xi - \mu)(\lambda-\zeta)  }
\end{eqnarray}
All these kernels act on $L^2(i\R)$.
\end{lemma}

%%%%%%%%%%%%%%%%%
With the convention that $\rho,\zeta,\xi\in i\R$ and $\lambda \in \gamma_R, \ \mu \in \gamma_L$, we have the following result.
\begin{lemma}
The operators in Lemma \ref{lemmaFourier} can be represented as the composition of several operators: 
\begin{gather}
\mathcal F(\pi K_\Ai\pi)(\xi, \zeta) = A(\xi,\mu) C(\mu, \lambda) B(\lambda, \zeta)\\
A(\xi,\mu) := \frac { {\rm e}^{(\mu-\xi) \wt \s - \frac {\mu^3}4 }}{2i\pi (\xi - \mu)} 
\ \ \ \ C(\mu,\lambda) := \frac { {\rm e}^{  \frac {\lambda^3 - \mu^3}{12}  }}{2i\pi (\mu - \lambda )}\ \ \ \ 
B(\lambda,\zeta) := \frac { {\rm e}^{\frac {\lambda^3}4 +  (\zeta - \lambda) \wt \s  }} {2i\pi ( \lambda-\zeta)} 
\end{gather}
\begin{gather}
\hline \nonumber \\
\mathcal F(\tilde \Pi K_\Ai^{(\tau,-\tau)}\tilde \Pi)(\xi,\zeta):= A _j(\xi, \mu) C (\mu,\lambda) B_k (\lambda, \zeta)\\
A_j(\xi, \mu)  := \sum_j\frac { (-1)^j {\rm e}^{(\xi- \mu)  \tilde a_j  -\frac {\mu^3}{12}+ \frac{\tau}{2^{2/3}} \mu^2 }}{2i\pi (\xi-\mu)\sqrt[3]{2} }\ \ \ 
 \
B_k(\lambda, \zeta):= \sum_k\frac { (-1)^k{\rm e}^{\frac {\lambda^3}{12}- \frac{ \tau}{2^{2/3}} \lambda^2 + (\lambda- \zeta) \tilde a_k} }{2i\pi ( \lambda-\zeta  )\sqrt[3]{2} }
\end{gather}
\begin{gather}
\hline \nonumber\\ 
\mathcal F(\tilde \Pi\mathfrak A_{\tau}\pi)(\xi,\zeta):= H_j(\xi,\lambda) Q_R (\lambda,\zeta)  -  A_j(\xi,\mu)C(\mu,\lambda) B(\lambda, \zeta) \\
H_j(\xi,\lambda) :=\sum_j (-1)^j \frac{ { \rm e}^{(\xi-\lambda)\tilde a_j - \wt \s\lambda  +\frac {\lambda^3}{12} + \frac{ \tau}{2^{2/3}} \lambda^2 } }{2i\pi (\xi - \lambda) \sqrt[3]{2} }\ ,\ \ \ \ 
Q_R(\lambda, \zeta) := \frac {{\rm e}^{\frac {\lambda^3}{12}  + \wt \s \zeta}}{2i\pi( \lambda - \zeta)}
\end{gather}
\begin{gather}
\hline \nonumber \\
\mathcal F(\pi\mathfrak A_{-\tau}^T\tilde \Pi)(\xi,\zeta):= Q_L (\xi,\mu) \wt H _k(\mu, \zeta) 
- A(\xi,  \mu) C (\mu,\lambda) B_k(\lambda, \zeta) \\
\wt H_k(\mu,\zeta) :=\sum_{k} (-1)^{k+1} \frac{ {\rm e}^{ (\mu-\zeta) \tilde a_k +  \mu \wt \s - \frac {\mu^3}{12} - \frac{ \tau}{2^{2/3}} \mu^2  }   }{2i\pi (\mu - \zeta) \sqrt[3]{2} }\ ,\ \ \ \ 
Q_L(\xi,\mu):= \frac {{\rm e}^{-\frac {\mu^3}{12} - \wt \s \xi } } {2i\pi (\xi - \mu)}
\end{gather}
with
\begin{eqnarray}
B, B_k: L^2(i\R)\to L^2(\gamma_R)\\
A, A_j: L^2(\gamma_L)\to L^2(i\R)\\
C: L^2(\gamma_R) \to L^2(\gamma_{L})   \\
H_j: L^2(\gamma_{R})\to L^2(i\R)\ \ \ Q_R:L^2(i\R)\to L^2(\gamma_R)\\
Q_L: L^2(\gamma_L) \to L^2(i\R) \  \ \ \wt H_k: L^2(i\R)\to L^2(\gamma_{L})
\end{eqnarray}
\end{lemma}

Finally,
\begin{prop}
The following identity of determinants holds
\begin{gather}
\det \le(\1 - \left[
\begin{array}{c|c}
AC B &   -  Q_L\wt H_k +  A C  B_k\\
\hline\\
  - H_j Q_R  + A_j C  B & A_j C  B_k
\end{array}
\ri]\ri)= 
\det  \le[
\begin{array}{c|c|c}
 \1 &  B &      B_k\\
 \hline
 AC & \1 & Q_L\wt H_k \\
\hline
 A_j C &H_j Q_R   & \1 
\end{array}
\ri] \nonumber \\
= \det \le[
\begin{array}{c|c|c|c|c|c}
\1_{L_1} & 0 & 0 &  0 & 0 & \wt H_k\\
\hline
0 & \1_{R_1} & 0 & 0 & Q_R & 0\\
\hline
0 & 0 & \1_{L_2} &  C & 0 & 0\\
\hline
0 & 0 & 0 & \1 & B & B_k\\
\hline
-Q_L & 0 & -A &  0 & \1_{0_2} & 0\\
\hline
0 & -H_j & -A_j &  0 & 0 & \1_{0_3}
\end{array}
\ri] = 
\det \le[
\begin{array}{c|c|c|c}
\1_{L_1} & \wt H_k H_j & \wt H_k A_j &0 \\
\hline
Q_R Q_L  & \1_{R_1} & Q_R A &0 \\
\hline
0 & 0 & \1_{L_2} & C \\
\hline
B Q_L  & B_k H_j  & B A + B_k A_j  & \1_{R_2} 
\end{array}
\ri] 
\end{gather}
where by the $\1_{X_j}$ we denote the identity operator on $L^2(X,\C)$ and the further subscript distinguishes orthogonal copies of the same space.
\end{prop}

\begin{proof}
We start by noticing that all operators are Hilbert--Schmidt, and hence the first two determinants and the last one are ordinary Fredholm determinants, since the operators appearing are trace-class; the third determinant should be understood as Carleman regularized $\det_2$ determinant. However, since the operator whose determinant is computed is diagonal-free, the formal definition coincides with the usual Fredholm determinant. The first identity is seen by multiplying on the left by a proper lower triangular matrix, while the second one is given by multiplying the matrix
\begin{gather}
\mathcal{M} = \le[
\begin{array}{c|c|c|c|c|c|c}
\1_{L_1} & 0 & 0 & 0 & 0 & 0 & \wt H_k\\
\hline
0 & \1_{R_1} & 0 & 0 & 0 & Q_R & 0\\
\hline
0 & 0 & \1_{L_2} & 0 & C & 0 & 0\\
\hline
0 & 0 & 0 & \1_{R_2} & 0 & B & B_k\\
\hline
0 & 0 & 0 & 0 & \1_{0_1} & 0 & 0\\
\hline
-Q_L & 0 & -A & 0 & 0 & \1_{0_2} & 0\\
\hline
0 &- H_j & -A_j & 0 & 0 & 0 & \1_{0_3}
\end{array}
\ri]
%= \cr
%=
%\le[
%\begin{array}{c|c|c|c|c|c|c}
%\1_{L_1} & 0 & 0 & 0 & 0 & 0 & \wt H_k\\
%\hline
%0 & \1_{R_1} & 0 & 0 & 0 & Q_R & 0\\
%\hline
%0 & 0 & \1_{L_2} & 0 & C_L & 0 & 0\\
%\hline
%0 & 0 & 0 & \1_{R_2} & 0 & -B & B_k\\
%\hline
%0 & 0 & 0 & 0 & \1_{0_1} &  C_R B  & -C_R B_k\\
%\hline
%0 & 0 & 0 & 0 & A C_L & \1_{0_2} & -Q_L \wt H_k\\
%\hline
%0 & 0 & 0 & 0 & -A_j C_L & -H_j Q_R & \1_{0_3}
%\end{array}
%\ri]
\end{gather}
on the left by
\begin{gather}
\mathcal{N} = \le[
\begin{array}{c|c|c|c|c|c|c}
\1_{L_1} & 0 & 0 & 0 & 0 & 0 & 0\\
\hline
0 & \1_{R_1} & 0 & 0 & 0 &  0 & 0\\
\hline
0 & 0 & \1_{L_2} & 0 & 0  & 0 & 0\\
\hline
0 & 0 & 0 & \1_{R_2} & 0 & 0 &  0 \\
\hline
0 & 0 & 0 & 0  & \1_{0_1} & 0 & 0\\
\hline
Q_L & 0 & A & 0 & 0 & \1_{0_2} & 0\\
\hline
0 & H_j & A_j & 0 & 0 & 0 & \1_{0_3}
\end{array}
\ri] 
\end{gather}
where $0_j$ is a copy of the imaginary axis $i\mathbb{R}$. We now multiply the two matrices in reverse order, as we know that $\det (\mathcal{M}\mathcal{N}) = \det(\mathcal{N}\mathcal{M})$.
%\begin{equation}
%\le[
%\begin{array}{c|c|c|c|c|c|c}
%\1_{L_1} & 0 & 0 & 0 & 0 & 0 &0\\
%\hline
%0 & \1_{R_1} & 0 & 0 & 0 &0 & 0\\
%\hline
%0 & 0 & \1_{L_2} & 0 &0 & 0 & 0\\
%\hline
%0 & 0 & 0 & \1_{R_2} & 0 &0& 0\\
%\hline
%0 & 0 & 0 & -C_R & \1_{0_1} & 0 & 0\\
%\hline
%Q_L & 0 & A & 0 & 0 & \1_{0_2} & 0\\
%\hline
%0 & 0 & H_j & A_j & 0 & 0 & \1_{0_3}
%\end{array}
%\ri]
%\end{equation}
In conclusion, we obtain the operator 
\bea
\det \le[
\begin{array}{c|c|c|c}
\1_{L_1} & \wt H_k H_j & \wt H_k A_j &0 \\
\hline
Q_R Q_L  & \1_{R_1} & Q_R A &0 \\
\hline
0 & 0 & \1_{L_2} & C \\
\hline
B Q_L  & B_k H_j  & B A + B_k A_j  & \1_{R_2} 
\end{array}
\ri]
\eea 
%(with the same determinant) (
where we have removed the trivial part involving the three copies of $i\R$.
\end{proof}

Collecting all the results found so far, we have
\begin{thm}
The gap probability of the tacnode process at single time is 
\begin{equation}
\det( \1 - \Pi \wt {\mathbb K} \Pi)  = F_2(\wt \s)^{-1} \cdot \det\left( \1 - \mathbb{M} \right)
\end{equation}
where
\begin{equation}
\mathbb{M} :=  \le[
\begin{array}{c|c|c|c}
0_{L_1} & - \wt  H_k H_j & - \wt H_k A_j &0 \\
\hline
- Q_R Q_L  & 0_{R_1} & - Q_R A &0 \\
\hline
0 & 0 &0_{L_2} &-C \\
\hline
- B Q_L  & - B_k H_j  & - (B A + B_k A_j)  &0_{R_2}
\end{array}
\ri]
\end{equation}
with
\begin{align}
 &Q_R Q_L(\lambda, \mu)  =\frac { { \rm e}^{\frac{\lambda^3 - \mu^3}{12}}}{2i\pi( \lambda - \mu  )}, \ \ \ \ 
 B Q_L(\lambda,\mu)  =\frac { {\rm e}^{\frac {\lambda^3}4 - \lambda\wt \s  - \frac { \mu^3}{12}} }{ 2i\pi ( \lambda - \mu)} \\
 &Q_R A(\lambda,\mu)  = \frac { {\rm  e}^{\mu \wt \s - \frac {\mu^3}{4} + \frac {\lambda^3}{12} } }{2i\pi (\lambda  - \mu)},  \ \ \  \ 
%BA(\lambda, \mu)    = \frac {{\rm e}^{ \frac {\lambda^3-\mu^3}4  +(\mu - \lambda) \wt \s }    }{2i\pi( \lambda - \mu )}  \\
C(\mu,\lambda) = \frac { {\rm e}^{  \frac {\lambda^3 - \mu^3}{12}  }}{2i\pi (\mu - \lambda )} \\
%&B_kA_j (\lambda, \mu)
%= \sum_{j=1}\frac { (-1)^{j}
% {\rm e}^{\frac {\lambda^3-\mu^3}{12}- \frac{ \tau}{2^{2/3}} (\lambda^2-\mu^2)  + (\lambda- \mu) a_j} }{2i\pi (\lambda -\mu)\sqrt[3]{4} }\\
&\wt H_k H_j(\mu,\lambda) =
%  \sum_{j=1} \frac{(-1)^{j+1}}{\sqrt[3]{4}} 
% \frac{ {\rm e}^{ (\mu -\lambda) a_j + (\mu - \lambda) \wt \s + \frac {\lambda^3 - \mu^3}{12} + \frac{ \tau}{2^{2/3}} (\lambda^2 - \mu^2 ) }  }{2i\pi (\mu- \lambda)  } \nonumber \\
% &  = 
  \frac  {\sum_{j=1}^{2K} (-1)^{j+1} h_j^{-1} (\mu) h_j(\lambda) } {2i\pi(\mu - \lambda)}\, ; \ \ \ \ h_j(\zeta):=\frac   { {\rm e}^{\zeta^3/12  + \frac{ \tau}{2^{2/3}} \zeta^2 - (\tilde a_j + \wt \s) \zeta } }{ \sqrt[3]{2}}\\
&\wt H_k A_j (\mu_1, \mu_2)  =
%\sum_{j=1}^{2K} (-1)^{j+1} \frac{ {\rm e}^{   \mu_1 \wt \s - \frac {\mu_1^3}{12} - \frac{ \tau}{2^{2/3}} \mu_1^2   -\frac {\mu_2^3}{12}+ \frac{ \tau}{2^{2/3}} \mu_2^2   }   }{2i\pi (\mu_1 - \mu_2) \sqrt[3]{4} }{\rm e}^{(\mu_2-\mu_1) a_j }  \\
%& = 
\sum_{j=1}^{2K} (-1)^{j+1} \frac {h_j^{-1}(\mu_1) g_j (\mu_2) }{2i\pi(\mu_1 -\mu_2)} \, ; \ \ \ \ g_j(\zeta):=  \frac {{\rm e}^{-\zeta^3/12 + \frac{ \tau}{2^{2/3}} \zeta^2 - \zeta \tilde a_j} }{ \sqrt[3]{2}} \\
&B_k H_j(\lambda_2, \lambda_1)  =
%\sum_{j=1}^{2K} \frac { (-1)^{j}{\rm e}^{\frac {\lambda_2^3}{12}- \frac{\tau}{2^{2/3}} \lambda_2^2 - \wt \s\lambda_1  +\frac {\lambda_1^3}{12} + \frac{ \tau}{2^{2/3}} \lambda_1^2 } 
% }{2i\pi (\lambda_2 - \lambda_1 )\sqrt[3]{4} } { \rm e}^{(\lambda_2-\lambda_1) a_j } 
 \sum_{j=1}^{2K} (-1)^{j} \frac {g^{-1}_j(\lambda_2) h_j(\lambda_1) }{2i\pi(\lambda_2 -\lambda_1)}   \\
 & \left(BA +B_kA_j\right)(\lambda,\mu)  = \frac { {\rm e}^{ \frac {\lambda^3-\mu^3}4  +(\mu - \lambda) \wt \s }    }{2i\pi(\lambda -  \mu )} +
  \sum_{j} (-1)^j \frac {g_j^{-1} (\lambda)g_j(\mu) }{2i\pi(\lambda -\mu)}
\end{align}
\end{thm}

\begin{proof}
The first three kernels and the kernel $BA$ follow from easy computations. 
\begin{gather}
 Q_R Q_L(\lambda, \mu) = \int_{i\R} \frac {\d \zeta}{2i\pi} \frac { { \rm e}^{\frac{\lambda^3 - \mu^3}{12}}}
   {2i\pi(\lambda - \zeta)( \zeta-\mu )} =\frac { { \rm e}^{\frac{\lambda^3 - \mu^3}{12}}}{2i\pi( \lambda - \mu  )}\\
 B Q_L(\lambda,\mu) = \int_{i\R} \frac {\d \zeta}{2i\pi} \frac { {\rm e}^{\frac {\lambda^3}4 - \lambda\wt \s  - \frac { \mu^3}{12}} }
  { 2i\pi ( \lambda-\zeta)( \zeta-\mu)} =\frac { {\rm e}^{\frac {\lambda^3}4 - \lambda\wt \s  - \frac { \mu^3}{12}} }{ 2i\pi ( \lambda - \mu)}\\
 Q_R A(\lambda,\mu) = \int_{i\R} \frac {\d \zeta}{2i\pi} \frac { {\rm  e}^{\mu \wt \s - \frac {\mu^3}{4} + \frac {\lambda^3}{12} } }{2i\pi (\lambda - \zeta)(\zeta  - \mu)} = \frac { {\rm  e}^{\mu \wt \s - \frac {\mu^3}{4} + \frac {\lambda^3}{12} } }{2i\pi (\lambda  - \mu)}  \\
BA(\lambda, \mu)  = \int_{i\R} \frac {\d \zeta}{2i\pi} \frac {{\rm e}^{ \frac {\lambda^3-\mu^3}4  +(\mu - \lambda) \wt \s }    }{2i\pi( \lambda - \zeta) (\zeta - \mu)}   = \frac {{\rm e}^{ \frac {\lambda^3-\mu^3}4  +(\mu - \lambda) \wt \s }    }{2i\pi( \lambda - \mu )}  
  \end{gather}

Next, we recall that the endpoints are ordered $\tilde a_j<\tilde a_{j+1}$, so that we can pick up residues accordingly to the sign of $\tilde a_j-\tilde a_k$ ($j,k=1, \ldots, 2K$). 
% for $j<k$ we close the integration on the right, for $j>k$ we close it on the left and for $j=k$ we can choose either:
\begin{gather}
\wt H_k H_j(\mu,\lambda) =\sum_{j,k} \frac{(-1)^{j+k+1}}{\sqrt[3]{4}}  \int_{i\R} \frac {\d \zeta}{2i\pi} {\rm e}^{\zeta(\tilde a_j -\tilde a_k)} 
 \frac{ {\rm e}^{ \mu \tilde a_k + \mu \wt \s - \frac {\mu^3}{12} - \frac{ \tau}{2^{2/3}} \mu^2  }   }{2i\pi (\mu - \zeta)  }
 \frac{ { \rm e}^{-\lambda \tilde a_j - \wt \s\lambda  +\frac {\lambda^3}{12} + \frac{ \tau}{2^{2/3}} \lambda^2 } }{ (\zeta - \lambda) }=
 \cr
 \sum_{j<k}  \frac{(-1)^{j+k}}{\sqrt[3]{4}} 
 \frac{ {\rm e}^{ (\mu -\lambda) \tilde a_k + (\mu - \lambda) \wt \s + \frac {\lambda^3 - \mu^3}{12} + \frac{ \tau}{2^{2/3}} (\lambda^2 - \mu^2 ) }  }{2i\pi (\mu - \lambda )  } 
 +  \sum_{k<j }  \frac{(-1)^{j+k}}{\sqrt[3]{4}} 
 \frac{ {\rm e}^{ (\mu -\lambda) \tilde a_j + (\mu - \lambda) \wt \s + \frac {\lambda^3 - \mu^3}{12} + \frac{ \tau}{2^{2/3}} (\lambda^2 - \mu^2 ) }  }{2i\pi (\mu - \lambda )  } + \cr
  + \sum_{j=1}^{2K} \frac{1}{\sqrt[3]{4}} 
 \frac{ {\rm e}^{ (\mu -\lambda) \tilde a_j + (\mu - \lambda) \wt \s + \frac {\lambda^3 - \mu^3}{12} + \frac{ \tau}{2^{2/3}} (\lambda^2 - \mu^2 ) }  }{2i\pi ( \mu-\lambda)  } 
 \end{gather}
%Since we have an even number of endpoints, the triangular sums not depending on one of the indices, we have cancellations every other term. From the first sum survive only the terms with {\bf even  $k$}, from the second sum survive only the terms with {\bf even $j$}, in both cases with a minus sign. This produces an alternating {\bf diagonal} sum with the last so that 
Thanks to some cancellations, we are left with
\begin{gather}
\wt H_k H_j(\mu,\lambda) =
  \sum_{j=1}^{2K} \frac{(-1)^{j+1}}{\sqrt[3]{4}} 
 \frac{ {\rm e}^{ (\mu -\lambda) \tilde a_j + (\mu - \lambda) \wt \s + \frac {\lambda^3 - \mu^3}{12} + \frac{ \tau}{2^{2/3}} (\lambda^2 - \mu^2 ) }  }{2i\pi (\mu- \lambda)  } .
\end{gather}

Similarly,
\begin{gather}
B_kA_j (\lambda, \mu):=\int_{i\R} \frac { {\rm d}\zeta}{2i\pi}  \sum_{k,j}\frac { (-1)^{k+j} {\rm e}^{\frac {\lambda^3}{12}- \frac{ \tau}{2^{2/3}} \lambda^2 + (\lambda- \zeta) \tilde a_k} }{2i\pi (\lambda -\zeta)\sqrt[3]{4} }\frac {{\rm e}^{(\zeta - \mu)  \tilde a_j  -\frac {\mu^3}{12}+ \frac{ \tau}{2^{2/3}} \mu^2 }}{(\zeta - \mu) }=\cr
= \sum_{j=1}^{2K} \frac { (-1)^{j}
 {\rm e}^{\frac {\lambda^3-\mu^3}{12}- \frac{ \tau}{2^{2/3}} (\lambda^2-\mu^2)  + (\lambda- \mu) \tilde a_j} }{2i\pi (\lambda -\mu)\sqrt[3]{4} }.
\end{gather}
%%%%%%%%%%%%%
In the next computation, we set $\lambda_1, \lambda_2\in \gamma_R$:
\begin{gather}
B_k H_j(\lambda_2, \lambda_1)  =\int_{i\R}  \frac {\d\zeta}{2i\pi} \sum_{k,j}\frac { (-1)^{k+j}{\rm e}^{\frac {\lambda_2^3}{12}- \frac{ \tau}{2^{2/3}} \lambda_2^2 + (\lambda_2- \zeta)\tilde a_k} }{2i\pi ( \lambda_2 -\zeta  )\sqrt[3]{4} } \frac{ { \rm e}^{(\zeta-\lambda_1) \tilde a_j - \wt \s\lambda_1  +\frac {\lambda_1^3}{12} + \frac{ \tau}{2^{2/3}} \lambda_1^2 } }{(\zeta - \lambda_1) } = \\=
\sum_{j\leq k}\frac { (-1)^{k+j}{\rm e}^{\frac {\lambda_2^3}{12}- \frac{ \tau}{2^{2/3}} \lambda_2^2 - \wt \s\lambda_1  +\frac {\lambda_1^3}{12} + \frac{ \tau}{2^{2/3}} \lambda_1^2 } 
 }{2i\pi (\lambda_2 - \lambda_1 )\sqrt[3]{4} } \le({ \rm e}^{(\lambda_2-\lambda_1) \tilde a_j } -{ \rm e}^{(\lambda_2-\lambda_1) \tilde a_k } \ri) 
\end{gather}
the first term contributes only with the terms with even $j$ (with positive sign) , the second only those with odd $k$ with a negative sign so that 
\begin{eqnarray}
B_k H_j(\lambda_2, \lambda_1)  =\sum_{j=1}^{2K} \frac { (-1)^{j}{\rm e}^{\frac {\lambda_2^3}{12}- \frac{ \tau}{2^{2/3}} \lambda_2^2 - \wt \s\lambda_1  +\frac {\lambda_1^3}{12} + \frac{ \tau}{2^{2/3}} \lambda_1^2 } 
 }{2i\pi (\lambda_2 - \lambda_1 )\sqrt[3]{4} } { \rm e}^{(\lambda_2-\lambda_1) \tilde a_j } 
\end{eqnarray}
Note that the kernel is regular at $\lambda_1 = \lambda_2$ because the sum vanishes.

In a similar way 
\begin{gather}
\wt H_k A_j (\mu_1, \mu_2)  =\int_{i\R} \frac{\d \zeta}{2i\pi} \sum_{k,j} (-1)^{k+j+1} \frac{ {\rm e}^{ (\mu_1-\zeta) \tilde a_k +  \mu_1 \wt \s - \frac {\mu_1^3}{12} - \frac{ \tau}{2^{2/3}} \mu_1^2  }   }{2i\pi (\mu_1 - \zeta) \sqrt[3]{4} }\frac { {\rm e}^{(\zeta- \mu_2)  \tilde a_j  -\frac {\mu_2^3}{12}+ \frac{ \tau}{2^{2/3}} \mu_2^2 }}{ (\zeta - \mu_2) }=\\
=\sum_{j\geq k} (-1)^{k+j+1} \frac{ {\rm e}^{   \mu_1 \wt \s - \frac {\mu_1^3}{12} - \frac{ \tau}{2^{2/3}} \mu_1^2  }   }{2i\pi (\mu_1 - \mu_2) \sqrt[3]{4} }\frac { {\rm e}^{  -\frac {\mu_2^3}{12}+ \frac{ \tau}{2^{2/3}} \mu_2^2 }}{ 1 }\le({\rm e}^{(\mu_1-\mu_2) \tilde a_k }  - {\rm e}^{(\mu_1-\mu_2) \tilde a_j }  \ri) = \\
=\sum_{j=1}^{2K} (-1)^{j+1} \frac{ {\rm e}^{   \mu_1 \wt \s - \frac {\mu_1^3}{12} - \frac{ \tau}{2^{2/3}} \mu_1^2   -\frac {\mu_2^3}{12}+ \frac{ \tau}{2^{2/3}} \mu_2^2   }   }{2i\pi (\mu_1 - \mu_2) \sqrt[3]{4} }{\rm e}^{(\mu_2-\mu_1) \tilde a_j }  
\end{gather}
\end{proof}

In conclusion, the kernel can be written as an integrable kernel in the sense of  Its-Izergin-Korepin-Slavnov (\cite{IIKS}):
%\begin{gather}
%-2i\pi (\xi - \zeta) \mathbb M(\xi,\zeta)   = \nonumber \\
%{\rm e}^{\zeta^3/12 -\xi^3/12} \chi_{_{L_2}}(\xi) \chi_{_{R_2}}(\zeta)+
% \le({\rm e}^{ \xi^3/4 - \xi \wt \s} \chi_{_{R_2}}(\xi)  + {\rm e}^{\xi^3/12} \chi_{_{R_1}}(\xi) \ri) 
% \le( {\rm e}^{-\zeta^3/12} \chi_{_{L_1}}(\zeta)  +  {\rm e}^{-\zeta^3/4+ \zeta \wt \s} \chi_{_{L_2}}(\zeta)
%\ri) +\nonumber \\
%+ \sum_{j=1}^{2K} (-1)^j 
%\le(g^{-1}_j(\xi) \chi_{_{R_2}} - h_j^{-1} (\xi) \chi_{_{L_1}}\ri)
%\le(g_j(\zeta) \chi_{_{L_2}} + h_j(\zeta) \chi_{_{R_1}}\ri)
%\end{gather}
%i.e. the kernel $\mathbb{M}$ is an integrable kernel in the sense of
\begin{equation}
\mathbb M(\xi,\zeta) = \frac{\textbf{f}(\xi)^T \cdot \textbf{g}(\zeta)}{2\pi i(\xi-\zeta)}
\end{equation}
with
\begin{align}
&\textbf{f}(\xi) = \le[
\begin{array}{c}
-{\rm e}^{-\xi^3/12} \chi_{_{L_2}}\\[3pt]
-{\rm e}^{ \xi^3/4 - \xi \wt \s} \chi_{_{R_2}} - {\rm e}^{\xi^3/12} \chi_{_{R_1}}\\[3pt]
g^{-1}_1(\xi) \chi_{_{R_2}} - h_1^{-1} (\xi) \chi_{_{L_1}}\\[3pt]
\vdots\\
-(-1)^{2K} g^{-1}_{_{2K}}(\xi) \chi_{_{R_2}} +(-1)^{2K} h_{_{2K}}  ^{-1} (\xi) \chi_{_{L_1}}
\end{array}
\ri] \\
&\textbf{g} (\zeta) = \le[
\begin{array}{c}
{\rm e}^{\zeta^3/12} \chi_{_{R_2}}\\[3pt]
 {\rm e}^{-\zeta^3/12} \chi_{_{L_1}}  + {\rm e}^{-\zeta^3/4+ \zeta \wt \s} \chi_{_{L_2}}\\[3pt]
g_1(\zeta) \chi_{_{L_2}} + h_1(\zeta) \chi_{_{R_1}}\\[3pt]
\vdots \\
g_{_{2K}}(\zeta) \chi_{_{L_2}} + h_{_{2K}}(\zeta) \chi_{_{R_1}}
\end{array}
\ri] 
\end{align}

It is thus natural to associate to it a suitable following RH problem. We refer to \cite{JohnSasha} for a detailed explanation.
%\begin{gather}
%\left\{ \begin{array}{cc}
%\Gamma_+ (\lambda) = \Gamma_-(\lambda) J (\lambda) & \lambda \in \Sigma \\
%\Gamma (\lambda) = I + \mathcal{O}\left( \lambda^{-1}\right) & \lambda \rightarrow \infty
%\end{array}
%\right. \\
%J(\lambda) := I - \textbf{f}(\lambda) \textbf{g}(\lambda)^T 
%\end{gather}
%where $\Sigma$ is the collection of all contours involved. 

\begin{prop}
The Fredholm determinant $\det(\Id-\mathbb{M})$ is linked through IIKS correspondence to the following Riemann-Hilbert problem
	\bea
		&&\Gamma_+(\lambda) = \Gamma_-(\lambda)J(\lambda), \quad \lambda \in \Sigma \nonumber\\
		&&\Gamma(\lambda) = I+\mathcal O(\lambda^{-1}),\quad \lambda\rightarrow\infty\label{RHH1}\\
		&&J(\lambda) :=  \label{mainRH}\\
		&& \left[\begin{array}{c|c|c|c|c|c}
			1 & {\rm e}^{-\Theta_{\wt \sigma}}\chi_{_L} & {\rm e}^{-\Theta_{\tau, -a_1}}\chi_{_L} & \ldots & \ldots & {\rm e}^{-\Theta_{\tau, -a_{2K}}}\chi_{_L}\\
			\hline
			{\rm e}^{\Theta_{\wt\sigma}}\chi_{_R} & 1 & {\rm e}^{\Theta_{-\tau, a_1}}\chi_{_R} & \ldots & \ldots & 
			{\rm e}^{\Theta_{-\tau, a_{2K}}}\chi_{_R}\\
			\hline
			-{\rm e}^{\Theta_{\tau, - a_1}}\chi_{_R} & {\rm e}^{-\Theta_{\tau, a_1}}\chi_{_L} & 1 & \ldots & \ldots & 0\\
			\hline
			\vdots & \vdots & 0 & \ldots & \ldots & \vdots\\
			\hline
			 \vdots & \vdots & 0 & \ldots & \ldots & \vdots\\
			 \hline
			(-1)^{2K}{\rm e}^{\Theta_{-\tau, -a_{2K}}}\chi_{_R} & (-1)^{2K+1}{\rm e}^{-\Theta_{\tau, a_{2K}}}\chi_{_L} & 0 & \ldots & \ldots & 1
		\end{array}\right] \nonumber
	\eea
where $\Sigma$ is the collection of all contours involved and with 
%$\Theta(\lambda;a_i) := \frac{\lambda^3}{6}-2^{-\frac{2}{3}}\tau\lambda^2-2^{-\frac{1}3}(a_i+\s)\lambda$.\\
\begin{gather}
\Theta_{\wt \sigma}(\lambda) = \frac{\lambda^3}{3} - \wt \sigma\lambda, \ \ \ \
\Theta_{\tau,a_i}(\lambda) = \frac{\lambda^3}{6} - 2^{-\frac{2}{3}}\tau \lambda^2 - 2^{-\frac{1}{3}}(a_i+\sigma)\lambda.
\end{gather} 
\end{prop}

\begin{proof}
It is simply a matter of straightforward calculations: using the standard formula $J(\lambda( = I - 2i \pi \textbf{f}(\lambda) \textbf{g}(\lambda)^T $ and writing explicitly the endpoints $\tilde a_i$ as functions of the original endpoints $a_i$, we get the jump matrix as in (\ref{mainRH}), but with two distinct copies of $\gamma_R$ and $\gamma_L$. On the other hand, it is easy to show that the jumps on - say - $\gamma_{R_1}$ and $\gamma_{R_2}$ commure, hence we can identify the two contours.
\end{proof}

In particular, let's consider the simplest case where $\mathcal{I} = [a,b]$, then the RH problem is $4\times 4$ with jump matrix
\begin{gather}\label{RHH}
%I - \textbf{f}(\lambda) \textbf{ g}(\lambda)^T = \nonumber \\
J(\lambda) =
 \left[ \begin{array}{cccc}
1 & e^{-\Theta_{\wt \sigma}} \chi_{_{L_1}} & e^{-\Theta_{\tau, -a}}\chi_{_{L_2}} & e^{-\Theta_{\tau, -b}}\chi_{_{L_2}} \\
e^{\Theta_{\wt \sigma}}\chi_{_{R_1}} & 1 & e^{\Theta_{-\tau, a}}\chi_{_{R_3}} & e^{\Theta_{-\tau, b}}\chi_{_{R_3}}  \\
-e^{\Theta_{\tau, -a}}\chi_{_{R_2}} & e^{-\Theta_{-\tau,a}}\chi_{_{L_3}} & 1 & 0 \\
e^{\Theta_{\tau,-b}}\chi_{_{R_2}} & -e^{-\Theta_{-\tau, b}}\chi_{_{L_3}} & 0 & 1
\end{array}
\right]  
\end{gather}
where
\begin{gather}
\Theta_{\wt \sigma}(\lambda) = \frac{\lambda^3}{3} - \wt \sigma\lambda, \ \ \ \
\Theta_{\tau,a_i}(\lambda) = \frac{\lambda^3}{6} - 2^{-\frac{2}{3}}\tau \lambda^2 - 2^{-\frac{1}{3}}(a_i+\sigma)\lambda.
\end{gather} 
(we have renamed the contours $R_1,R_2,R_3$ and $L_1,L_2,L_3$).

We will now focus exclusively on the single-interval case and we will apply a steepest descent method in order to prove the factorization of the gap probability of Tacnode process into two gap probabilities of the Airy process. The starting point is the $4\times 4$ Riemann-Hilbert problem (\ref{RHH}) with contour configuration as in \figurename \ \ref{curvesigma} or \figurename \ \ref{curvetau}, depending on the scaling regime we are considering.

\section{Proof of Theorem \ref{THMsigma}}\label{Asymsigma}
From now on, we are assuming $\tau>0$. For $\tau\leq0$ the calculations follow the same guidelines as below.
%, with some \textit{caveat} due to the change of sing or to the vanishing of $\tau$. 
% the critical points switch places and with some caveat it's possible to follow the same arguments. For $\tau=0$ I wrote a few lines below.

%\paragraph{Phases.}
The phase functions $\Theta_{\tau}(\lambda, -b)$ and $\Theta_{-\tau}(\lambda, a)$ (appearing in the entries of the $2\times 2$ off-diagonal blocks of the jump matrix (\ref{RHH}))  have inflection points with zero derivative when the discriminant of the derivative vanishes, which occurs when
\begin{gather}
a_{\rm crit} + \sigma + \tau^2 =0 \\
b_{\rm crit} - \sigma -\tau^2 =0
\end{gather}
with critical values $\Theta_{\tau}(\lambda, -b_{\rm crit}) = 2^{1/3}\tau$ and $\Theta_{-\tau}(\lambda, a_{\rm crit}) = -2^{1/3}\tau$. The neighbourhood of the discriminant is parametrizable as follows
\begin{gather}
a = a(t) = -\sigma -\tau^2 +t \\
b = b(s) = \sigma + \tau^2 - s
\end{gather}
Thus we have the expressions
\begin{align}
&\Theta_{\tau}(\lambda, -b) = \frac{\xi_-^3}{3} - s\xi_- + \frac{\tau^3}{3} - s\tau, \ \ \  \xi_-:= \frac{\lambda  - 2^{\frac{1}{3}}\tau}{2^{\frac{1}{3}}} \\
&\Theta_{-\tau}(\lambda, a) = \frac{\xi_+^3}{3} - t\xi_+ - \frac{\tau^3}{3} +t\tau,  \ \ \ \xi_+:= \frac{\lambda  + 2^{\frac{1}{3}}\tau}{2^{\frac{1}{3}}}
\end{align}

On the other hand, the phase $\Theta_{\wt \sigma}$ in the entries $(1,2)$ and $(2,1)$ of (\ref{RHH}) has critical point at $\pm \sqrt{\wt \sigma} = \pm \sqrt{2^{\frac{2}{3}} \sigma}$.

\paragraph{Preliminary step.}
We conjugate the matrix $\Gamma$ by the constant (with respect to $\lambda$) diagonal matrix
\begin{equation}
D := \operatorname{diag}(1,1, -K(t), -K(s)) \label{Dmatrix}
\end{equation}
where $K(u) := \frac{\tau^3}{3} -u\tau$. As a result, also the jump matrices (\ref{RHH}) are similarly conjugated and this has the effect of replacing the phases $\Theta_{\pm \tau, \mp a}$ and $\Theta_{\pm \tau, \mp b}$ by  ``$\Theta_{\pm \tau, \mp a} \mp K(t)$" and ``$\Theta_{\pm \tau, \mp b} \mp K(s)$" respectively, so that their critical value is zero.

We denote by a hat the new matrix and respective jump
\begin{equation}
\hat \Gamma := e^{-D} \Gamma e^D \ \ \ \hat J := e^{-D} J e^D
\end{equation}

Thus, the resulting jump $\hat J$ has the following form:
%\subsection{RH setting}
%We are in the following situation (see also Tacnode\_Airy-sigma.mw; for an illustrative picture, see flying papers) after conjugation with the matrix $\text{exp} \{\text{diag}(1,1, K(s), K(t))\}$
\begin{gather}
\left[ \begin{array}{cccc}
1& e^{-\Theta_{\wt\sigma}} & 0 &0 \\
0&1&0&0 \\
0&0&1&0 \\
0&0&0&1 
\end{array} \right]  \ \  \text{on} \ L_1,  \ \ \ \left[ \begin{array}{cccc}
1&0&0&0 \\
e^{\Theta_{\wt \sigma}} &1&0&0 \\
0&0&1&0 \\
0&0&0&1 
\end{array} \right] \ \ \text{on} \ R_1, \label{RHPconj1}
\end{gather}
\begin{gather}
\left[ \begin{array}{cccc}
1&0&e^{- \Theta_{\tau, -a} +K(t)}& e^{- \Theta(\xi_-,s)} \\
0&1&0&0 \\
0&0&1&0 \\
0&0&0&1 
\end{array} \right]  \ \  \text{on} \ L_2, \ \ \ 
\left[ \begin{array}{cccc}
1&0&0&0 \\
0&1&0&0 \\
-e^{\Theta_{\tau,-a}-K(t)}&0&1&0 \\
e^{ \Theta(\xi_-,s)}&0&0&1 
\end{array} \right]  \ \  \text{on} \ R_2,
\end{gather}
\begin{gather}
\left[ \begin{array}{cccc}
1&0&0&0 \\
0&1&0&0 \\
0&e^{- \Theta(\xi_+,t)}&1&0 \\
0&-e^{- \Theta_{-\tau,b}-K(s)}&0&1 
\end{array} \right]  \ \  \text{on} \ L_3, \ \ \ 
\left[ \begin{array}{cccc}
1&0&0&0 \\
0&1&e^{\Theta(\xi_+,t)}& e^{ \Theta_{-\tau,b}+K(s)} \\
0&0&1&0 \\
0&0&0&1 
\end{array} \right]  \ \  \text{on} \ R_3, 
\end{gather}
%and
%\begin{gather}
%J_2 := L_1 R_2 L_1^{-1} = \left[ \begin{array}{cccc}
%1&0&0&0 \\
%0&1&0&0 \\
%- e^{\tau^3 (\Theta(-\tau,-a)+K(s))}&e^{-\tau^3 \Theta(\xi_-,s)}&1&0 \\
%e^{\tau^3 \Theta(\xi_+,t)}&-e^{-\tau^3 (\Theta(\tau,b)-K(t))}&0&1 
%\end{array} \right] = R_1 L_3 R_1^{-1} =: J_3
%\end{gather}
where
\begin{gather}
\Theta(\xi_{\pm}, u):= \frac{\xi_{\pm}^3}{3} - \xi_\pm u \ \ \ \ \ \xi_\pm:= \frac{\lambda \pm \sqrt[3]{2}\tau}{\sqrt[3]{2}}  \label{xi+-} \\
\Theta_{-\tau,b}(\lambda,s) := \frac{\lambda^3}{6} + \frac{\tau \lambda^2}{2^{2/3}} -2^{2/3}\sigma \lambda- \frac{\tau^2\lambda}{\sqrt[3]{2}}  + \frac{s\lambda}{\sqrt[3]{2}}\\
\Theta_{\tau, -a}(\lambda, t) := \frac{\lambda^3}{6} - \frac{\tau \lambda^2}{2^{2/3}}  -2^{2/3}\sigma \lambda - \frac{\tau^2\lambda}{\sqrt[3]{2}} +\frac{t\lambda}{\sqrt[3]{2}} \label{RHPconjfin}
\end{gather}
%keeping in mind $\Theta(-\tau,-a) + \Theta(\tau,a) = \Theta $; critical values
%\begin{align}
%a &= -\sigma -\tau^2 +  s &s\in \mathbb{R}\\
%b &= \sigma + \tau^2 -  t &t\in \mathbb{R}
%\end{align}

We choose the contours according to the following configuration (see \figurename \ \ref{curvesigma}):
\begin{itemize}
\item $L_2$ and $R_2$ are centred around the critical point $P_R :=  2^{\frac{1}{3}}\tau$ 
\item $L_3$ and $R_3$ are centred around the critical point $P_L :=- 2^{\frac{1}{3}}\tau$
\item $L_1$ passes through the critical point $P_{\sigma,L} := - \sqrt{\wt \sigma}$ and $R_1$ passes through the critical point $P_{\sigma,R}:=\sqrt{\wt \sigma}$; these points are thought as very far from the origin, in the limit as $\sigma \gg 1$.
\end{itemize}

\begin{figure}
\centering
%\resizebox{.7\textwidth}{!}{
%\input{curvesigma_t}
%}
\includegraphics[width=.8\textwidth]{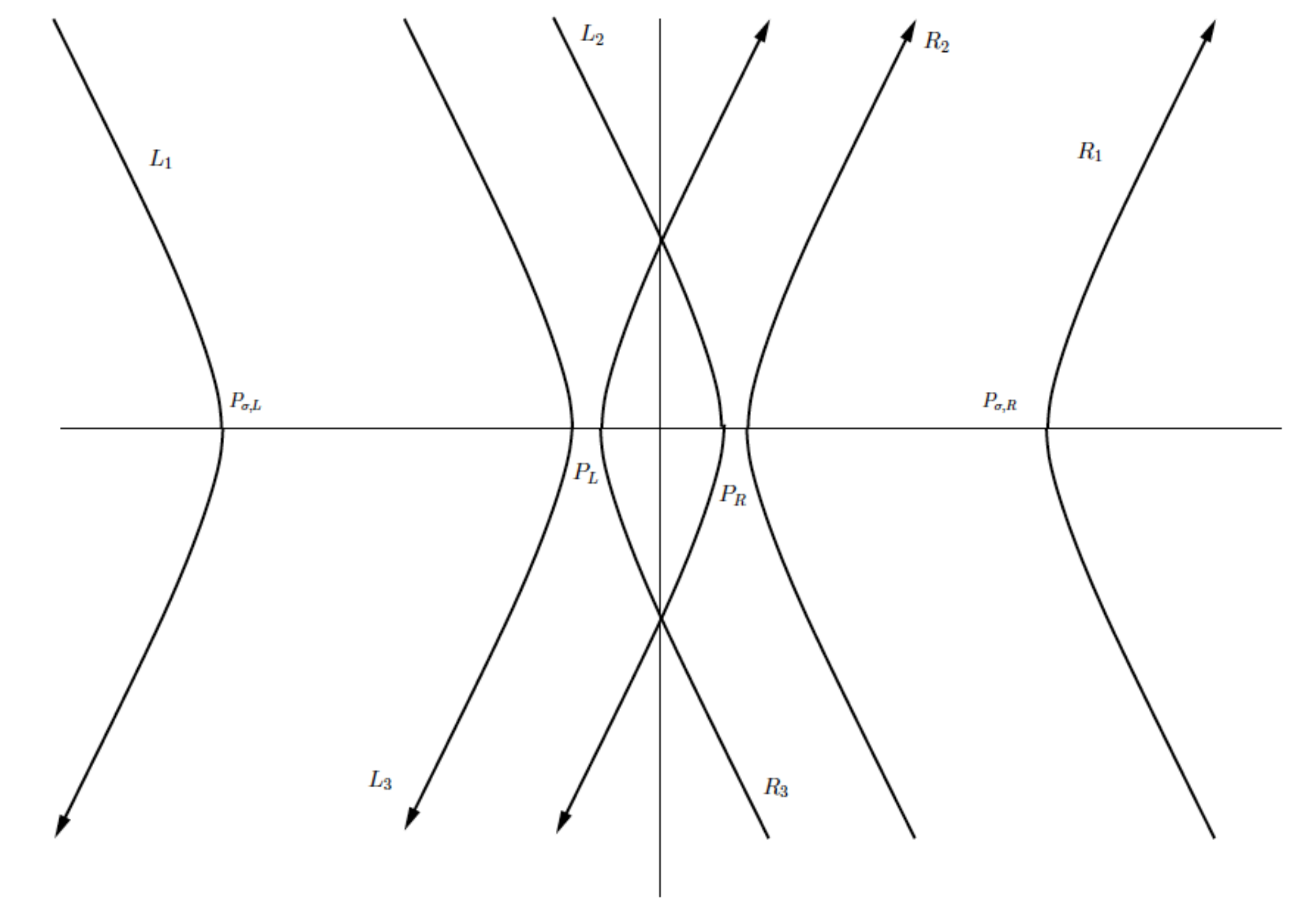}
\caption{The contour setting in the asymptotic limit as $\sigma \rightarrow +\infty$.}
\label{curvesigma}
\end{figure}

\begin{oss}
All the left jumps commute with themselves and similarly all the right jumps. Moreover, the jump matrices $L_2$ and $R_3$ commute.
\end{oss}

%The main idea of the proof is now that as $\sigma \rightarrow + \infty$ :
The proof now proceeds along the following scheme (as $\sigma \rightarrow +\infty$):
\begin{enumerate}
\item the matrices $L_1$ and $R_1$ are exponentially close to the identity in every $L^p$ norm (Lemma \ref{lemmaLR1});
\item regarding the matrices $L_2$ and $R_2$, the entries of the form $\pm(\Theta_{\tau,-a} - K(t))$ are exponentially small in every $L^p$ norm; the same behaviour will appear for the entries of the type $\pm (\Theta_{-\tau,b}+K(s))$ in the matrices $L_3$ and $R_3$ (Lemma \ref{lemmaLR23});
\item  for the remaining entries in the jumps $L_{2,3}$ and $R_{2,3}$ we will explicitly and exactly solve a (model) Riemann-Hilbert problem which will approximate the problem at hand.
\end{enumerate}

\subsection{Estimates on the phases}
%The following calculations can be found on "Tacnode-rescaling-sigma.mw"

The proof of the first two points rely on the following lemmas.

\begin{lemma}\label{lemmaLR1}
The jumps on the curves $L_1$ and $R_1$ are exponentially suppressed in any $L^p$ norm, $1\leq p \leq \infty$, as $\sigma \rightarrow + \infty$.
\end{lemma}

\begin{proof}
A parametrization for the curves $L_1$ and $R_1$ is the following $\lambda = \pm 2^{1/3}\sqrt{\sigma} +u\left(\frac{1}{2}\pm \frac{\sqrt{3}}{2}i\right)$. 
%passes thorugh the critical points $\pm 2^{1/3}\sqrt{\sigma}$ with parametrization 
Therefore, we have (for both signs)
\begin{gather}
\Re \left[ \Theta_{\tilde \sigma; R_1} \right] = \Re \left[ - \Theta_{\tilde \sigma; L_1}\right] = -\frac{4}{3}\sigma^{3/2} -\frac{\sqrt{\sigma}}{2^{2/3}} u^2- \frac{u^3}{3}
\end{gather}
which implies
\begin{gather}
\left\| e^{\Theta_{\tilde \sigma}} \right\|^p_{L^p(R_1)} =  2 \int_0^\infty e^{p \Re \left[ \Theta_{\tilde \sigma} \right]}du  \leq Ce^{-\frac{4}{3}p\sigma^{3/2}}, \ \ 
\left\| e^{\Theta_{\tilde \sigma}} \right\|_{L^\infty(R_1)} = e^{-\frac{4}{3}\sigma^{3/2} }
\end{gather}
The same results holds for the contour $L_1$.
\end{proof}

\begin{lemma}\label{lemmaLR23}
Given $0<K_1<1$ fixed and $s<K_1(\sigma+\tau^2)$, then the function $e^{ \Theta(-\tau,b)+K(s)}$ tends to zero exponentially fast in any $L^p(R_3)$ norm ($1\leq p\leq \infty$) as $\sigma \rightarrow + \infty$:
\begin{equation}
\left\| e^{\Theta(-\tau,b)+K(s)}\right\|_{L^p(R_3)} \leq Ce^{-2\tau(1-K_1)\sigma}
\end{equation}
Similarly, the function $e^{- \Theta(-\tau,b)-K(s)}$ is exponentially small in any $L^p(L_3)$ norm  ($1\leq p\leq \infty$). 

Moreover, the function $e^{-\Theta_{\tau,-a}+K(t)}$ and $e^{\Theta_{\tau,-a}-K(t)}$ are exponentially small in any $L^p(L_2)$ and $L^p(R_2)$ norms, respectively ($1\leq p\leq \infty$).
\end{lemma}

\begin{proof}
A parametrization of $R_3$ is $\lambda = \sqrt[3]{2}\tau + u\left[ \frac{1}{2} \pm \frac{2}{\sqrt{3}}i  \right]$, $u\geq 0$. This yields
\begin{gather}
\Re \left[ \Theta(-\tau,b) +K(s) \right] = -\frac{u^3}{6} -\frac{\delta u}{2^{\frac{4}{3}}} -2\tau \sigma-2{\tau}^{3}+2\tau \delta
%-\frac{u^3}{6} -\frac{u}{\sqrt[3]{2}}\left( \sigma +{\tau}^{2} - \frac{s}{2}\right) -2\tau \sigma-2{\tau}^{3}+2\tau t
\end{gather}
where we set $s = 2\sigma + 2\tau^2 - \delta$, $0< \delta < \sigma + \tau^2$, and this is valid for both branches of the curve.

Regarding the $L^p(R_3)$ norms, we have that $\left| e^{\Theta_{-\tau,b}+K(s)}\right| = e^{\Re[\Theta_{-\tau,b}+K(s)] }$; therefore,
\begin{gather}
\left\| e^{ \Theta(-\tau,b)+K(s)}\right\|^p_{L^p(R_3)} \leq   2Ce^{- 2p\tau(\sigma+ \tau^2-\delta)} \left[\int_0^1 e^{- 2^{-\frac{4}{3}} p \delta u }du + \int_1^\infty e^{-p\frac{u^3}{6}}du  \right] \nonumber \\
%\leq Ce^{-2p\tau(1-K_1)(\sigma+\tau^2)}\left[ \mathcal{O}\left(\frac{1}{\sigma} \right)  + C\right] 
\leq Ce^{-2p\tau (1-K_1)\sigma}\\
\left\| e^{\Theta(-\tau,b)+K(s)}\right\|_{L^\infty(R_3)} = e^{-2\tau(\sigma+\tau^2-\delta)}\leq Ce^{-2\tau(1-K_1)\sigma}
\end{gather}
given that $s<K_1(\sigma+\tau^2)$ with $0<K_1<1$.

All the other cases are completely analogous.
\end{proof}

\subsection{Global parametrix. The model problem}
In this subsection we will use the Hasting-McLeod matrix (see \cite{Ptrans}, but in the normalization of \cite{MeM}) as parametrix for the RH problem related to $\hat \Gamma$.

Let us consider the following model problem: 
\begin{equation}
\left\{ 
\begin{array}{ll}
\Omega_+(\lambda) = \Omega_-(\lambda) J_R (\lambda) & \text{on} \ L_2\cup R_2 \\ 
\Omega_+(\lambda) = \Omega_-(\lambda) J_L (\lambda) & \text{on} \ L_3\cup R_3 \\
\Omega(\lambda) = I + \mathcal{O}\left( \lambda^{-1}\right)& \text{at} \ \infty
\end{array}
\right.
\end{equation}
with jumps (see \figurename \ \ref{RHmodel})
\begin{gather}
J_R:= \left[ \begin{array}{cccc}
1&0& 0 & e^{- \Theta(\xi_-,s)}\chi_{_{L_2}} \\
0&1&0&0 \\
 0 &0&1&0 \\
 e^{ \Theta(\xi_-,s)} \chi_{_{R_2}} &0&0& 1 
\end{array} \right] \\
J_L:= \left[ \begin{array}{cccc}
1&0&0&0 \\
0& 1& e^{\Theta(\xi_+,t)}\chi_{_{R_3}}&0 \\
0& e^{- \Theta(\xi_+,t)}\chi_{_{L_3}}& 1&0 \\
0&0&0&1 
\end{array} \right] 
\end{gather} 
and we recall $\xi_{\pm}:= \frac{\lambda \pm \sqrt[3]{2} \tau}{\sqrt[3]{2}}$ as defined in (\ref{xi+-}).

\begin{figure}
\centering
%\resizebox{.8\textwidth}{!}{
%\input{RHmodel_t}
%}
\includegraphics[width=.8\textwidth]{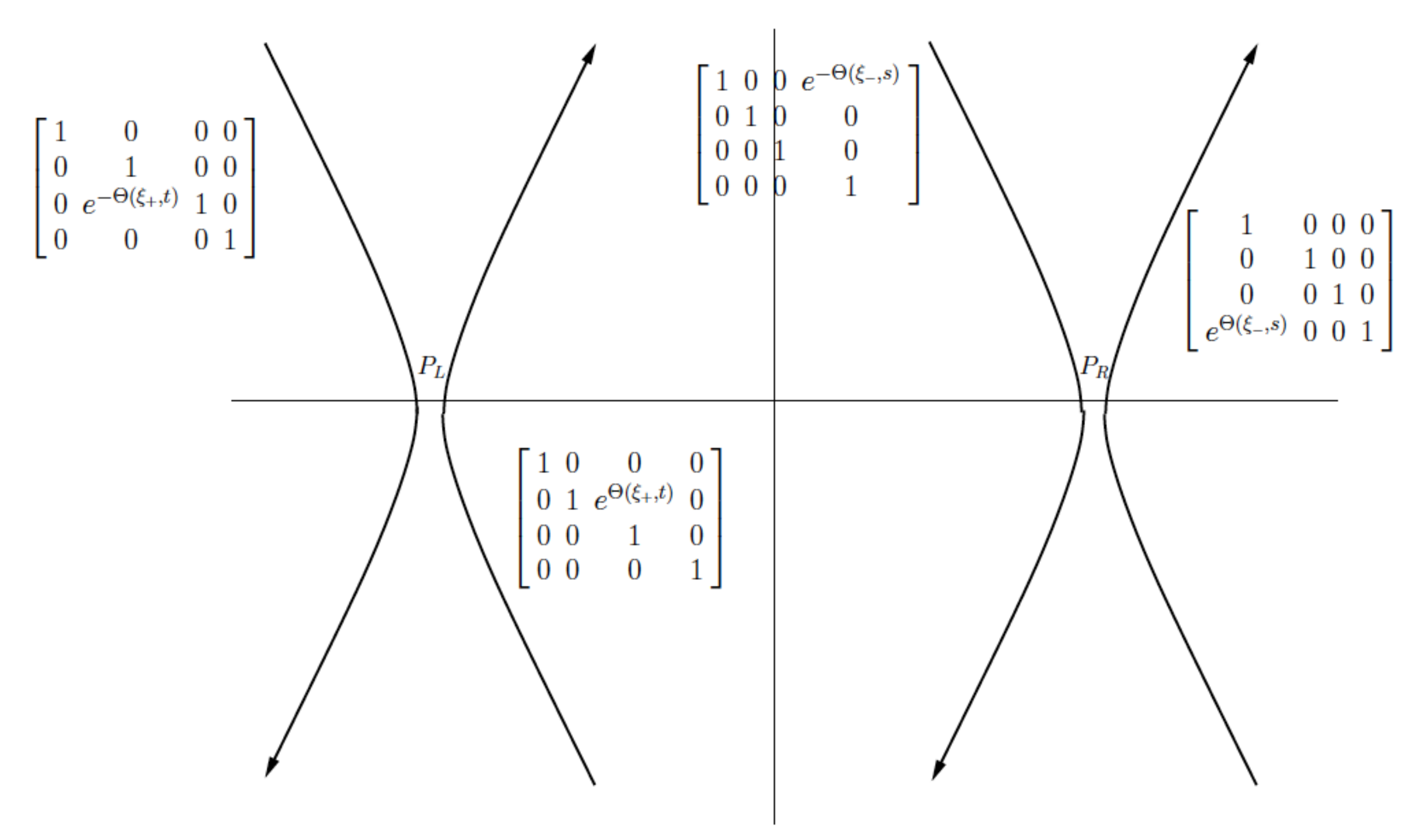}
\caption{The contour setting with the jump matrices in the model problem.}
\label{RHmodel}
\end{figure}

This model problem can be solved in exact form by considering two solutions of the Hasting-McLeod Painlev\'e II RH problem, namely 
\begin{equation}
\Phi_{HM}(s) \ \ \ \text{and} \ \ \ \tilde \Phi_{HM}(t):= \sigma_3\sigma_2 \Phi_{HM}(t) \label{HMLsolution}\sigma_2\sigma_3, 
\end{equation}
where $\sigma_2, \, \sigma_3$ are Pauli matrices and $\Phi_{HM}(u)$  is the solution to a $2\times 2$ RH problem with jump matrix
\begin{equation}
\left[ \begin{array}{cc}
1 & e^{\Theta(\lambda, u)}\chi_{_{\gamma_R}} \\
e^{-\Theta(\lambda,u)}\chi_{_{\gamma_L}}&1
\end{array}
\right] , \ \ \ \Theta(\lambda, u) = \frac{\lambda^3}{3} - u\lambda
\end{equation}
and behaviour at infinity normalized to the identity $2\times 2$ matrix; as usual, $\gamma_R$ is a contour which extends to infinity along the rays $\arg (\lambda) = \pm \frac{i\pi}{3}$ and $\gamma_L = -\gamma_R$ (for more details see \cite{MeM}).
The asymptotic behaviour of the functions (\ref{HMLsolution}) as $\xi \rightarrow \infty$ is
\begin{align}
\Phi(\xi_+,t ) &= I + \frac{1}{\xi_+}\left[ \begin{array}{cc} p(t)& q(t) \\ -q(t) & -p(t) \end{array}  \right] + \mathcal{O}\left( \frac{1}{\xi_+^2} \right) \\
\tilde \Phi(\xi_-, s)&=\sigma_3\sigma_2\left[  I + \frac{1}{\xi_-}\left[ \begin{array}{cc} p(s)& q(s) \\ -q(s) & -p(s) \end{array}  \right] + \mathcal{O}\left( \frac{1}{\xi_-^2} \right) \right] \sigma_2\sigma_3 \nonumber \\
&= I + \frac{1}{\xi_-}\left[ \begin{array}{cc} -p(s)& -q(s) \\ q(s) & p(s) \end{array}  \right] + \mathcal{O}\left( \frac{1}{\xi_-^2} \right)
\end{align}

The global parametrix, i.e. the exact solution of the model problem, is then easily verified to be given by
\begin{equation}
\Omega:= \left[ \begin{array}{cccc}
\tilde \Phi_{11}(\xi_-, s)&0&0&\tilde \Phi_{12}(\xi_-,s) \\
0& \Phi_{11}(\xi_+,t)& \Phi_{12}(\xi_+,t)&0 \\
0&\Phi_{21}(\xi_+,t)&\Phi_{22}(\xi_+,t)&0 \\
\tilde \Phi_{21}(\xi_-,s)&&0& \tilde \Phi_{22}(\xi_-,s)
\end{array} \right] .
\end{equation}

\subsection{Approximation and error term for the matrix $\hat \Gamma$}

The following relation holds
\begin{equation}
\hat \Gamma = \mathcal{E} \cdot \Omega
\end{equation}
where $\mathcal{E}$ is the ``error" matrix. The goal is to show that the RHP satisfied by the error matrix has jump equal to a small perturbation of the identity matrix $I + \mathcal{O}(\sigma^{-\infty})$, so that the Small Norm Theorem can be applied (see \cite[Appendix C]{MeM}).

\begin{lemma}\label{lemmaerror}
Given $s,t < K_1(\sigma + \tau^2)$ with $0<K_1<1$, the error matrix $\mathcal{E} = \hat \Gamma(\lambda) \Omega^{-1}(\lambda)$ solves a RH problem with jumps on the contours as indicated in \figurename \ \ref{curvesigma} and of the following orders
\begin{gather}
\left\{ 
\begin{array}{ll}
\mathcal{E}_+(\lambda) = \mathcal{E}_-(\lambda) J_\mathcal{E}(\lambda) & \text{on} \ \Sigma \\
\mathcal{E}(\lambda) = I + \mathcal{O}\left(\lambda^{-1}\right) & \text{as} \ \lambda \rightarrow \infty
\end{array} \right. \\
J_\mathcal{E} =  
\left[ \begin{array}{cccc}
1 &  \mathcal{O} (\sigma^{-\infty})\chi_{_{L_1}}  & \mathcal{O} (\sigma^{-\infty})\chi_{_{L_2}} & 0 \\
 \mathcal{O} (\sigma^{-\infty})\chi_{_{R_1}}  &1& 0 &\mathcal{O}(\sigma^{-\infty})\chi_{_{R_3}} \\
 -\mathcal{O}(\sigma^{-\infty}) \chi_{_{R_2}}& 0 &1&0 \\
0 & -\mathcal{O}(\sigma^{-\infty})\chi_{_{L_3}} &0& 1 
\end{array} \right] 
\end{gather}
and the $\mathcal{O}$-symbols are valid in any $L^p$ norms ($1\leq p \leq \infty$).
\end{lemma}

\begin{proof}
First of all, we notice that, thanks to Lemma \ref{lemmaLR1} and \ref{lemmaLR23}, all the extra phases that were not included in the model problem $\Omega$ behave like $\mathcal{O}(\sigma^{-\infty})$ as $\sigma \rightarrow \infty$ in any $L^p$ norm. The jump of the error problem are  the remaining jumps appearing in the original $\hat \Gamma$-problem conjugated with the Hasting-McLeod solution $\Omega$, which is independent on $\sigma$:
\begin{gather}
J_{\mathcal{E}} = \Omega^{-1} \hat J \Omega = \nonumber \\
 \Omega^{-1} \left[ \begin{array}{cccc}
1& e^{-\Theta_{\wt \sigma}}\chi_{_{L_1}} &e^{- \Theta_{\tau, -a} + K(t)}\chi_{_{L_2}}& 0 \\
e^{\Theta_{\wt \sigma}}\chi_{_{R_1}}&1&0& e^{ \Theta_{-\tau,b}+K(s)}\chi_{_{R_3}} \\
-e^{\Theta_{\tau,-a}-K(t)}\chi_{_{R_2}}&0&1&0 \\
0&-e^{- \Theta_{-\tau,b}-K(s)}\chi_{_{L_3}}&0&1 
\end{array} \right] \Omega \nonumber \\
= \Omega^{-1} \left( I + \mathcal{O}(\sigma^{-\infty}) \right) \Omega = I  +\mathcal{O}(\sigma^{-\infty})
\end{gather}
since $\Omega$ and $\Omega^{-1}$ are uniformly bounded in $\sigma$.
\end{proof}

We recall that the Small Norm Theorem says that uniformly on closed sets not containing the contours of the jumps
\begin{equation}
\left\| \mathcal{E} (\lambda) - I\right\| \leq \frac{C}{\text{dist}(\lambda, \Sigma)} \left( \left\|J_{\mathcal{E}} -I \right\|_1 + \frac{\left\|J_{\mathcal{E}} -I \right\|^2_2}{1-\left\|J_{\mathcal{E}} -I \right\|_\infty} \right) 
\end{equation}
where $\Sigma$ is the collection of all contours. Thanks to Lemma \ref{lemmaerror}, we conclude
\begin{equation}
\left\| \mathcal{E} (\lambda) - I\right\| \leq \frac{C}{\text{dist}(\lambda, \Sigma)} e^{-K\sigma}
\end{equation}
for some positive constants $C$ and $K$. The error matrix $\mathcal{E}$ is then found as the solution to the integral equation
\begin{equation}
\mathcal{E}(\lambda) = I + \int_\Sigma \frac{\mathcal{E}_-(w) \left(J_\mathcal{E}(\lambda) - I \right) \, dw}{2\pi i (w-\lambda)}
\end{equation}
and can be obtained by iterations
\begin{equation}
\mathcal{E}^{(0)}(\lambda) = I, \ \ \ \mathcal{E}^{(k+1)}(\lambda) =  I + \int_\Sigma \frac{\mathcal{E}^{(k)}_-(w) \left(J_\mathcal{E}(\lambda) - I \right) \, dw}{2\pi i (w-\lambda)}
\end{equation}
and, thanks to Lemma \ref{lemmaerror} we have
\begin{equation}
\mathcal{E}(\lambda) = I + \frac{1}{\text{dist} (\lambda, \Sigma)}\mathcal{O} \left(\sigma ^{-\infty} \right).
\end{equation}

\subsection{Conclusion of the proof of Theorem \ref{THMsigma}}

We recall here a main theorem about Fredholm determinants of IIKS integrable kernels (see  \cite{Misomonodromic} and \cite{MeM})

\begin{thm} \label{thmtaufunctionresiduessigma}
The Fredholm determinant $\det (\operatorname{Id} - \hat \Pi \mathbb{H} \hat \Pi)$ of (\ref{explode0}) satisfies the following differential equations
%is  equal to the Jimbo-Miwa-Ueno isomonodromic $\tau$-function of the RH problem (\ref{RHH}). For any parameter $\rho$ on which the integral operator $\hat\Pi \mathbb{H} \hat \Pi$ may depend, we have
\begin{gather}
\partial_{\rho} \ln \det (\operatorname{Id} - \hat \Pi \mathbb{H} \hat \Pi) = \omega_{JMU} (\partial_\rho) = \int_\Sigma \operatorname{Tr} \left( \Gamma_-^{-1}(\lambda) \Gamma'_-(\lambda) \partial_\rho \Xi(\lambda)\right) \frac{d\lambda}{2\pi i} \label{JMUtauH}
%\partial_{t} \ln \det (\operatorname{Id} - \hat \Pi \mathbb{H} \hat \Pi) = \omega_{JMU} (\partial_t) = \int_\Sigma \operatorname{Tr} \left( \Gamma_-^{-1}(\xi) \Gamma'_-(\xi) \partial_t \Xi(\xi)\right) \frac{d\xi}{2\pi i}
\end{gather}
More specifically,
\begin{align}
\partial_s \ln \det (\operatorname{Id}- \hat \Pi \mathbb{H}\hat \Pi) &= -  \operatorname{res}_{\lambda = \infty} \operatorname{Tr} \left(\Gamma^{-1}\Gamma'\partial_s T \right) = \frac{1}{\sqrt[3]{2}\lambda}\Gamma_{1; \, (4,4)} \\
\partial_t \ln \det (\operatorname{Id} - \hat \Pi\mathbb{H}\hat \Pi)&= -  \operatorname{res}_{\lambda = \infty} \operatorname{Tr} \left(\Gamma^{-1}\Gamma'\partial_t T \right)  = - \frac{1}{\sqrt[3]{2}\lambda}\Gamma_{1; \, (3,3)} 
\end{align}
where $\Gamma_1 := \lim_{\lambda \rightarrow \infty} \lambda (\Gamma(\lambda) - I)$. 
\end{thm}

\begin{proof}
We notice that the original RHP for $\Gamma$ (see (\ref{RHH})) is equivalent to a RH problem with constant jumps up to a conjugation with the matrix
\begin{gather}
T = \text{diag} \left[ \frac{\kappa}{4}, -\Theta_{\wt \sigma}+\frac{\kappa}{4}, - \Theta_{\tau,-a}+\frac{\kappa}{4}, -\Theta_{\tau,-b}+\frac{\kappa}{4} \right] \label{conjT} \\
\kappa = \Theta_{\wt \sigma} + \Theta_{\tau,-a} + \Theta_{\tau,-b}
\end{gather}
Thus, the matrix $\Psi:= \Gamma e^T$ solves a RHP with constant jumps and it is (sectionally) a solution to a polynomial ODE. 

Applying the Theorem \cite[Theorem 2.1]{MeM} to the case at hand, we have the equality (\ref{JMUtauH}). Moreover, using the Jimbo-Miwa-Ueno residue formula, we can explicitly calculate
\begin{gather}
\partial_s \ln \det (\operatorname{Id}- \hat \Pi \mathbb{H}\hat \Pi) = -  \operatorname{res}_{\lambda = \infty} \operatorname{Tr} \left(\Gamma^{-1}\Gamma'\partial_s T \right) \\
\partial_t \ln \det (\operatorname{Id} - \hat \Pi\mathbb{H}\hat \Pi)= -  \operatorname{res}_{\lambda = \infty} \operatorname{Tr} \left(\Gamma^{-1}\Gamma'\partial_t T \right)   
\end{gather}
Taking into account the asymptotic behaviour at $\infty$ of the matrix $\Gamma$ we have
\begin{gather}
\text{Tr} \left[ \Gamma^{-1}\Gamma'\partial_sT \right] = \text{Tr} \left[ \left( -\frac{\Gamma_1}{\lambda^2}  + \mathcal{O}\left( \lambda^{-3}\right)\right) \left( \frac{\partial_s \kappa}{4}I - \partial_s \Theta_{\tau,-b} E_{4,4} \right) \right] 
=  - \frac{1}{\sqrt[3]{2}\lambda}\Gamma_{1; \, (4,4)} \\
\text{Tr} \left[ \Gamma^{-1}\Gamma'\partial_tT \right] = \text{Tr} \left[ \left( -\frac{\Gamma_1}{\lambda^2}  + \mathcal{O}\left( \lambda^{-3}\right)\right)  \left( \frac{\partial_t \kappa}{4}I - \partial_t \Theta_{\tau,-a} E_{3,3} \right) \right]  
= + \frac{1}{\sqrt[3]{2}\lambda}\Gamma_{1; \, (3,3)}
\end{gather}
since $\det \Gamma \equiv 1$ which implies $\text{Tr}\ \Gamma_1 =0$.
\end{proof}

We now use the exact formula in Theorem \ref{thmtaufunctionresiduessigma} to conclude the proof of Theorem \ref{THMsigma}; recall that 
\begin{equation}
\Gamma(\lambda) = e^{D} \mathcal{E}(\lambda) \Omega(\lambda) e^{-D}
\end{equation}
and thanks to Lemma \ref{lemmaerror} we have
\begin{gather}
\Gamma_1 = e^{D} \hat \Gamma_1 e^{-D} = \Omega_1 \left(I + \mathcal{O}(\sigma^{-\infty}) \right) \nonumber \\
= \sqrt[3]{2} \left[ \begin{array}{cccc}
-p(s)&0&0& -q(s) \\
0&p(t)& q(t)&0 \\
0& -q(t)& -p(t)&0 \\
q(s)&&0&p(s)
\end{array} \right] \left(I + \mathcal{O}(\sigma^{-\infty}) \right)
\end{gather}
%\begin{oss}
%The coefficient $\sqrt[3]{2}$ pops out because the Airy RHP's above ($\tilde \Phi$ and $\Psi$) were defined for the variable $\xi_{\pm} := \frac{\lambda \pm \sqrt[3]{2}\tau}{\sqrt[3]{2}}$, thus as $\lambda \rightarrow \infty$ $\xi_{\pm}\sim 2^{-1/3}\lambda$.
%\end{oss}
which yields
\begin{align}
\Gamma_{1; \, (4,4)} &=  \Omega_{1; \, (4,4)} = \sqrt[3]{2} p(s) +  \mathcal{O}(\sigma^{-\infty})  \\
\Gamma_{1; \, (3,3)} &= \Omega_{1; \, (3,3)}  = -\sqrt[3]{2}p(t) + \mathcal{O}(\sigma^{-\infty}). 
\end{align}
Recall that $p(u)$ is the logarithmic derivative of the gap probability for the Airy process (i.e the Tracy-Widom distribution); collecting all the previous results, we have
\begin{gather}
\text{d}_{s,t} \ln \det \le(\operatorname{Id}  - \mathbb{H}\bigg|_{[-\sigma -\tau^2 +t, \sigma + \tau^2 -s ]}\ri)  \nonumber \\
=p(s) \text{ds}+ p(t) \text{dt} + \mathcal{O}\left( \sigma^{-\infty}\right) \text{ds} + \mathcal{O}\left( \sigma^{-\infty}\right) \text{dt} + \mathcal{O}\left( \sigma^{-\infty}\right) \text{ds dt} 
\end{gather}
uniformly in $s,t $ within the domain that guarantees the uniform validity of the estimates above as per Lemma \ref{lemmaerror}, namely, $s,t < K_1 (\sigma + \tau^2)$, $0<K_1<1$.

We now integrate from $(s_0, t_0)$ to $(s,t)$ with $s_0:= a+ \sigma + \tau^2$, $t_0 = -b+\sigma + \tau^2$ and we get
\begin{gather}
%\ln \det  (\operatorname{Id}  - \mathbb{H}\chi_{[-\sigma -\tau^2 +s, \sigma + \tau^2 -t ]}) - \ln \det  (\operatorname{Id}  - \mathbb{H}\chi_{[a, b ]}) \nonumber \\
\ln \det  \le(\operatorname{Id}  - \mathbb{H}\bigg|_{[-\sigma -\tau^2 +t, \sigma + \tau^2 -s ]}\ri) \nonumber \\
%- \ln \det  (\operatorname{Id}  - \mathbb{H}\chi_{[-\sigma -\tau^2 +s_0, \sigma + \tau^2 -t_0 ]}) 
%= \int_{s_0}^s p(u) du + \int_{t_0}^t p(w) dw + {\color{red} \text{things that are supposed to be small...}} \nonumber \\
=  \ln \det \le(\operatorname{Id}  - K_{\Ai}\bigg|_{[s,+\infty)}\ri) +  \ln \det \le(\operatorname{Id}  - K_{\Ai}\bigg|_{[t,+\infty}\ri) + \mathcal{O}(\sigma^{-1})  + C \label{asymptoticsigma}
\end{gather}
with $C = \ln \det  \le(\operatorname{Id}  - \mathbb{H}\bigg|_{[a,b]}\ri) $.

In conclusion,
\begin{gather}
\det \le(\operatorname{Id} - \mathbb K^{\rm tac}\bigg|_{[-\sigma -\tau^2 +t, \sigma + \tau^2 -s ]}\ri) \nonumber \\
= \frac{e^C\det \le(\operatorname{Id}  - K_{\Ai}\bigg|_{[s,+\infty)}\ri)\det \le(\operatorname{Id}  - K_{\Ai}\bigg|_{[t,+\infty}\ri) \left( 1 + \mathcal{O}(\sigma^{-1})\right) }{\det \le(\operatorname{Id}  - K_{\Ai}\bigg|_{[\wt \s,\infty)}\ri)}
\end{gather}
On the other hand, the Fredholm determinant of the Airy kernel appearing in the denominator tends to unity as $\sigma \rightarrow \infty$, thus we only need to prove that the constant $C$ is zero. Indeed this is the case

%\subsection{The case $\tau=0$.} Everything works also in the case $\tau=0$, with some adjustments. 

%In particular, we have that the critical points collapse to $0$, thus the curves $L_2/L_3$ and $R_2/R_3$ coincide. The phases which are not supposed to be critical, I am planning to leave them on the curves $L_1/R_1$ which will flee at $\infty$ as $\sigma \rightarrow + \infty$. 

%In this way, the bothering phases will go to zero and I can built up the solution of the ``model problem" assembling two Airy solutions as before, but this time they apply in the same point (this shouldn't be a problem, should it? No, it's not). 

%Check the file Tacnode\_Airy-sigma.mw for the jump matrices and Tacnode-rescaling-sigma.mw for the negative phases along the curves. 

%I have the feeling that this approach is actually working for all $\tau$, not only $\tau=0$, but I'm leaving everything this way, because it's too much work and right now it's still working.

\begin{lemma}\label{Cequal0}
The constant of integration $C$ in (\ref{asymptoticsigma}) is zero.
\end{lemma}

\begin{proof}

We recall the definition of the integral operator $\hat \Pi \mathbb{H} \hat \Pi$ acting on $\mathcal{H}_1\oplus \mathcal{H}_2 = L^2([\wt \sigma, \infty)) \oplus L^2([\tilde a,\tilde b])$, with kernel 
\begin{equation}
\hat \Pi \mathbb{H} \hat \Pi = \left[ \begin{array}{c|c}
\pi K_{\Ai} \pi & -\sqrt[6]{2} \pi \mathfrak{A}^T_{-\tau}\tilde \Pi \\ \hline
-\sqrt[6]{2} \tilde \Pi \mathfrak{A}_\tau \pi & \sqrt[3]{2}\tilde \Pi K^{(\tau,-\tau)}_{\Ai} \tilde \Pi
\end{array}
\right]
\end{equation}
where $\hat \Pi := \pi \oplus \tilde \Pi$, $\pi$ is the projector on $[\wt \sigma, +\infty)$, $\Pi$ is the projector on $[\tilde a,\tilde b]$ and
\begin{gather}
K_{\Ai}(x,y):= \int_{0}^\infty \text{Ai}(x+u)\text{Ai}(y+u) \, du \\
K_{\Ai}^{(\tau,-\tau)}(\sigma-x, \sigma-y) :=  e^{\tau(y-x)} \times \nonumber \\
\times \int_0^\infty du\,  \text{Ai}(\sigma-x+\tau^2+\sqrt[3]{2}u)\text{Ai}(\sigma-y+\tau^2+\sqrt[3]{2}u)  \\
\mathfrak{A}_{\tau} (x,y):= \,  \text{Ai}^{( \tau)}(x-\sigma + \sqrt[3]{2}y) - \int_0^\infty \text{Ai}^{( \tau)}(\sigma-x+\sqrt[3]{2}v) \text{Ai}(v+y)\, dv \nonumber \\
= \, 2^{1/6} e^{\tau \left(x-\sigma + \sqrt[3]{2}y\right) + \frac{2}{3}\tau^3} \text{Ai}(x-\sigma + \sqrt[3]{2}y+\tau^2) + \nonumber \\
  -2^{1/6}\int_0^\infty  dv \, e^{\tau \left( \sigma-x+\sqrt[3]{2}v\right)+\frac{2}{3}\tau^3 } \text{Ai}(\sigma-x+\sqrt[3]{2}v+\tau^2) \text{Ai}(v+y)\\
\mathfrak{A}^T_{-\tau}(x,y) :=  
%\, \text{Ai}^{(-\tau)}(y-\sigma + \sqrt[3]{2}x) - \int_0^\infty \text{Ai}^{(-\tau)}(\sigma-y+\sqrt[3]{2}v) \text{Ai}(v+x), dv \nonumber \\
2^{1/6} e^{-\tau \left(y-\sigma + \sqrt[3]{2}x\right) - \frac{2}{3}\tau^3} \text{Ai}(y-\sigma + \sqrt[3]{2}x+\tau^2) + \nonumber \\
  -2^{1/6}\int_0^\infty  dv \, e^{-\tau \left( \sigma-y+\sqrt[3]{2}v\right) - \frac{2}{3}\tau^3 } \text{Ai}(\sigma-y+\sqrt[3]{2}v+\tau^2) \text{Ai}(v+x).
\end{gather}

We would like to perform some uniform pointwise estimates on the entries of the kernel in order to prove that as $\sigma\rightarrow +\infty$ the trace of the operator $\hat \Pi \mathbb{H} \hat \Pi $ tends to zero.

Indeed, 
\begin{gather}
|\pi K_{\Ai}(u,v) \pi| \leq \frac{C_1}{\sqrt{\sigma}} e^{-\frac{2}{3}u^{3/2} - \frac{2}{3}v^{3/2}}\\
|\sqrt[3]{2}\Pi K^{(\tau,-\tau)}_{\Ai}(x,y)\tilde \Pi| \leq C_2 e^{-\sigma^{3/2}} \\
|\sqrt[6]{2}\tilde \Pi \mathfrak{A}_{\tau}(x,v) \pi| \leq C_3 e^{-\tau^2 \sqrt{\sigma}} e^{\tau\left( \sqrt[3]{2}v-\sigma \right)-\frac{2}{3}\left( \sqrt[3]{2}v-\sigma +\tilde a \right)^{3/2}} \\
|\sqrt[6]{2}\pi \mathfrak{A}_{-\tau}^T(u,y)\tilde \Pi| \leq C_4 e^{-\tau^2 \sqrt{ \sigma}}  e^{-\tau \left(\sqrt[3]{2}u-\sigma \right) -\frac{2}{3}\left( \sqrt[3]{2}u-\sigma +\tilde a\right)^{3/2}} 
\end{gather}
for some positive constants $C_j$ ($j=1,\ldots,4$), where we used the convention that $x,y$ are the variables running in $[\tilde a,\tilde b]$ and $u,v$ are the variables running in $[\tilde \sigma, \infty)$. Such estimates follow from simple arguments on the asymptotic behaviour of the Airy function when its argument is very large.

Collecting all the estimates, we get
\begin{gather}
\left[ \begin{array}{c|c}
\pi K_{\Ai} \pi & -\sqrt[6]{2} \pi \mathfrak{A}^T_{-\tau}\tilde \Pi \\ \hline
- \sqrt[6]{2} \tilde \Pi \mathfrak{A}_\tau \pi &\tilde \Pi K^{(\tau,-\tau)}_{\Ai} \tilde \Pi
\end{array}
\right] \leq C_\sigma \left[ \begin{array}{c|c} f(u) f(v) & f(u) \\  \hline f(v)  & 1  \end{array}\right]
\end{gather}
with $ C_{\sigma} = \frac{\max \{C_j, \  j=1,\ldots, 4 \}}{\sqrt{\sigma}}$ and $f(z) = e^{\tau \left(\sqrt[3]{2}u-\sigma \right) -\frac{2}{3}\left( \sqrt[3]{2}u-\sigma -2^{-1/3}\sigma +2^{-1/3}a\right)^{3/2}}$. 
On the right hand side we have a new operator $\mathcal{L}$ acting on the same Hilbert space $L^2([\wt \sigma, \infty)) \oplus L^2([\tilde a,\tilde b])$ with trace
\begin{equation}
\operatorname{Tr} \mathcal{L} =    \left\| f \right\|_{L^2(\wt \sigma, \infty)}^2 +(\tilde b-\tilde a) \leq C(b-a) \label{traceL}
\end{equation}
for some positive constant $C$, since $\left\| f \right\|_{L^2(\wt \sigma, \infty)}^2 \rightarrow 0$ as $\sigma \rightarrow + \infty$. 

Concluding, keeping $[a,b]$ fixed,
\begin{gather}
|\ln \det (\operatorname{Id}  - \hat \Pi \mathbb{H}\hat \Pi)| =  \sum_{n=1}^\infty \frac{\text{Tr} \, (\hat \Pi \mathbb{H}\hat \Pi^n)}{n} \nonumber \\
 \leq  \sum_{n=1}^\infty \frac{C_\sigma^n (b-a)^n}{n} \leq \frac{C_\sigma (b-a)}{1-C_\sigma (b-a)} \rightarrow 0
\end{gather}
as $\sigma \rightarrow + \infty$. This implies that the constant of integration $C$ must be zero.
\end{proof}

\section{Proof of Theorem \ref{THMtau}}\label{Asymtau}

We deal now with the case $\tau\rightarrow \pm \infty$, i.e. we are moving away from the tacnode point along the boundary curves of the domain so that there is one of the gaps that divaricates as we proceed. From now on, we will only focus on the case $\tau \rightarrow + \infty$. The case $\tau\rightarrow -\infty$ is analogous. 
%(the critical point again switches, thus the curves as well).
%The situation of the curves is the one depicted in Tacnode\_Airy-tau.mw and Tacnode-rescaling-tau.mw. 

The RH problem we are considering is the same as for the proof of Theorem \ref{THMsigma} (\ref{RHH1})-(\ref{RHH}). We conjugates the jumps with the constant diagonal matrix $D$ (see definition (\ref{Dmatrix})) and we have the same jump matrices as in (\ref{RHPconj1})-(\ref{RHPconjfin}).

The position of the curves is depicted in \figurename \ \ref{curvetau}:
\begin{itemize}
\item $L_2$ and $R_2$ are centred around the critical point $P_R :=  2^{\frac{1}{3}}\tau$ 
\item $L_3$ and $R_3$ are centred around the critical point $P_L :=- 2^{\frac{1}{3}}\tau$
\item $L_1$ passes through the critical point $P_{\sigma,L} := - \sqrt{\wt \sigma}$ and $R_1$ passes through the critcal point $P_{\sigma,R}:=\sqrt{\wt \sigma}$.
\end{itemize}
The points $P_{R/L} = \pm 2^{\frac{1}{3}}\tau$ are thought as very far from the origin, in the limit as $\tau \gg 1$.

We need to perform certain ``contour deformations''  and "jump splitting" in the RHP (\ref{RHH1})-(\ref{RHH}). 
To explain these manipulation consider a general RHP  with a jump on a certain contour $\gamma_0$ and with jump matrix $J(\lambda)$ 
\begin{equation}
\Gamma_+(\lambda) = \Gamma_-(\lambda) J(\lambda)\ ,\ \ \lambda \in \gamma_0.
\end{equation}
The ``contour deformation'' procedure stands for the following; suppose  $ \gamma_1$ is another contour  such that
\begin{itemize}
\item $\gamma_0\cup  \gamma_1^{-1}$  is the positively oriented boundary of  a domain $D_{\gamma_0, \gamma_1}$, where $\gamma_1^{-1}$ stands for the contour traversed in the opposite orientation,
\item $J(\lambda)$ and $J^{-1}(\lambda)$ are both analytic in  $D_{\gamma_0, \gamma_1}$ and  (in case the domain extends to infinity) $J(\lambda) \to I + \mathcal O(\lambda^{-1})$ as $|\lambda| \to\infty, \lambda \in D_{\gamma_0, \gamma_1}$.
\end{itemize}
We define $\wt \Gamma(\lambda) = \Gamma(\lambda)$ for $\lambda \in \C \setminus D_{\gamma_0,\gamma_1}$ and $\wt \Gamma(\lambda) = \Gamma(\lambda) J(\lambda)^{-1} $ for $\lambda \in  D_{\gamma_0,\gamma_1}$. This new matrix then has jump on $ \gamma_1$ with jump matrix $J(\lambda)$ ($\lambda \in \gamma_1$) and no jump (i.e. the identity jump matrix) on $\gamma_0$. While technically this is a new Riemann Hilbert problem, we shall refer to it with simply as the ``deformation'' of the original one, without introducing a new symbol.

The ``jump splitting'' procedure stands for a similar manipulation: suppose that the jump matrix relative to the contour $\gamma$  is factorizable into two (or more) matrices $J(\lambda) = J_0(\lambda) J_1(\lambda)$.   Let $\wt \gamma$, $D_{\gamma,\wt \gamma}$ be exactly as in the description above. Then define $\wt \Gamma(\lambda) = \Gamma(\lambda)$ for $\lambda \in \C \setminus D_{\gamma,\wt \gamma}$ and 
$\wt \Gamma(\lambda) = \Gamma(\lambda) J_+(\lambda)^{-1} $ for $\lambda \in  D_{\gamma,\wt \gamma}$. Then $\wt \Gamma$ has jumps 
$$
\wt \Gamma_+(\lambda) = \wt \Gamma_-(\lambda) J_0(\lambda),  \ \ \lambda\in \gamma_0 ,\ \ 
\wt \Gamma_+(\lambda) = \wt \Gamma_-(\lambda) J_1(\lambda),  \ \ \lambda\in \gamma_1 .
$$
Also in this case, while this is technically a different RHP, we shall refer to it with the same symbol $\Gamma$. We will also refer to the inverse operation as ``jump merging".

\begin{figure}
\centering
%\resizebox{.8\textwidth}{!}{
%\input{curvetau_t}
%}
\includegraphics[width=.8\textwidth]{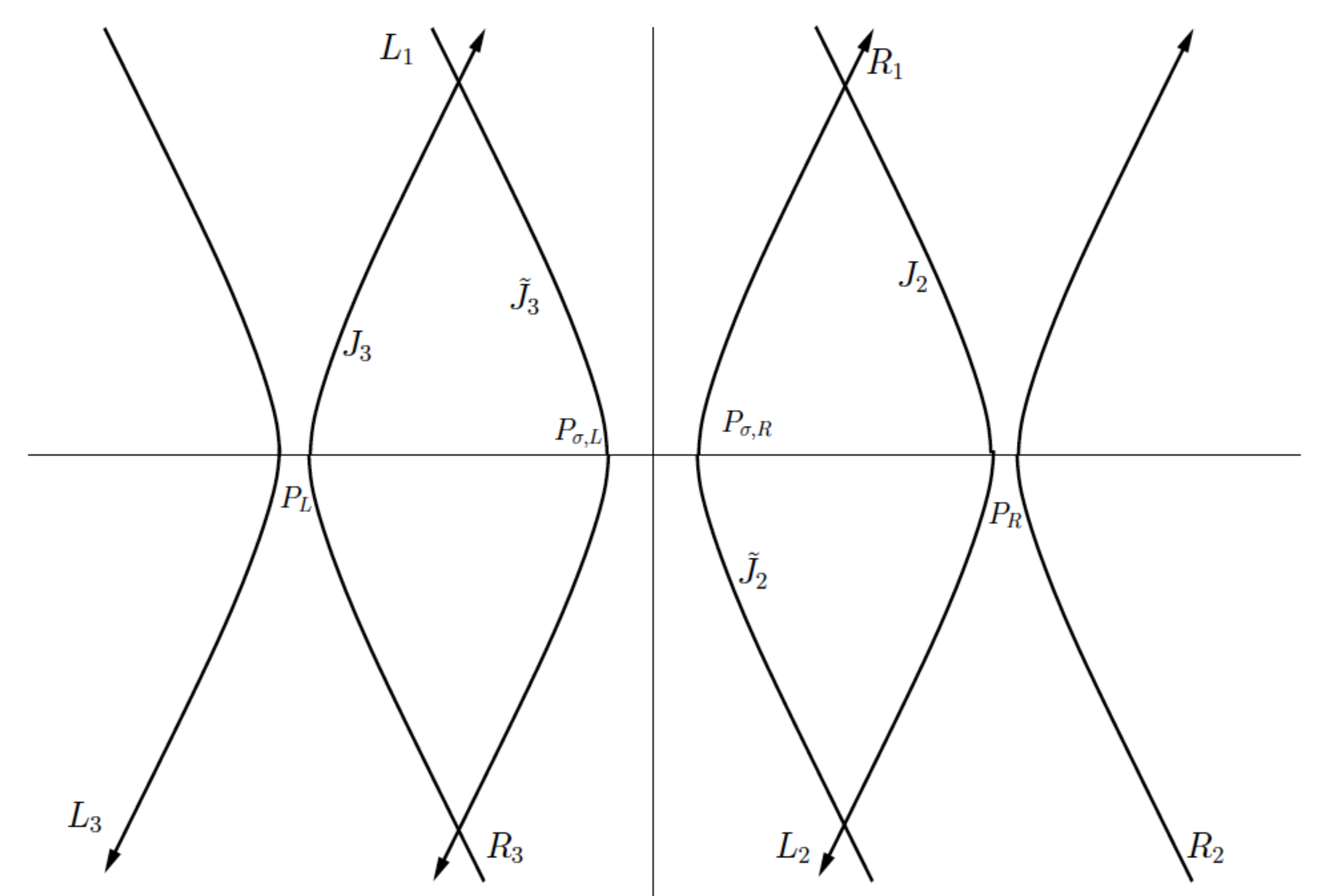}
\caption{The contour setting in the asymptotic limit as $\tau \rightarrow +\infty$.}
\label{curvetau}
\end{figure}

With this terminology in mind, we deform $R_3$ on the left next to its critical point $-\sqrt[3]{2}\tau$ leads to a new jump matrix on $R_3$, due to conjugation with the curve $L_1$ (similarly for $L_2$)
% on the right next to its critical point $+\sqrt[3]{2}\tau$.
%  and after passing over the curves $L_1/R_1$ (stuck at the origin) we get the matrices
\begin{gather}
J_3 := L_1 R_3 L_1^{-1} =  R_1 L_2 R_1^{-1} =: J_2 \nonumber \\
= \left[ \begin{array}{cccc}
1&0&e^{-\Theta(\tau,-a)+K(t)}&e^{-\Theta(\xi_-,s)} \\
0&1&e^{\Theta(\xi_+,t)}&e^{\Theta(-\tau,b)+K(s)} \\
0&0&1&0 \\
0&0&0&1 
\end{array} \right]
\end{gather}

Again as before, the proof is based on estimating the phases in the jump matrices which are not critical and solving the RH problem by approximation with an exact solution to a model problem.

\subsection{Estimates of the phases}

First of all we notice that a similar version of Lemma \ref{lemmaLR1} does not apply here, since the phases on the contours $L_1$ and $R_1$ do not depend on $\tau$. On the other hand, we can partially restate Lemma \ref{lemmaLR23} applied to the case at hand when $\tau\rightarrow \infty$.
\begin{lemma}\label{lemmaLR23tau}
Given $0<K_1<1$ fixed and $s<K_1(\sigma+\tau^2)$, then the function $e^{ \Theta(-\tau,b)+K(s)}$ tends to zero exponentially fast in any $L^p(R_3)$ norm ($1\leq p\leq \infty$) as $\tau \rightarrow + \infty$:
\begin{equation}
\left\| e^{\Theta(-\tau,b)+K(s)}\right\|_{L^p(R_3)} \leq Ce^{-2(1-K_1)\tau^3}
\end{equation}
Similarly, the functions $e^{- \Theta(-\tau,b)-K(s)}$, $e^{\Theta_{\tau,-a}-K(t)}$ and $e^{-\Theta_{\tau,-a}+K(t)}$ are exponentially small in any $L^p(L_3)$, $L^p(R_2)$ and $L^p(L_2)$ norms, respectively ($1\leq p\leq \infty$). 
%Similarly, the function $e^{-\Theta_{-\tau,-a}-K(s)}$ is exponentially small in any $L^p(L_2)$ norm  ($1\leq p\leq \infty$). 
%Moreover, the function $e^{-\Theta_{-\tau,-a}-K(s)}$ and $e^{\Theta_{-\tau,-a}+K(s)}$ are exponentially small in any $L^p(L_2)$ and $L^p(R_2)$ norms, respectively ($1\leq p\leq \infty$).
\end{lemma}
\begin{proof}
Using the same parametrization as in Lemma \ref{lemmaLR23}, we have
\begin{gather}
\left\| e^{ \Theta(-\tau,b)+K(t)}\right\|^p_{L^p(R_3)} \leq   2Ce^{- 2p\tau(\sigma+ \tau^2-\delta)} \left[\int_0^1 e^{- 2^{-\frac{4}{3}}\delta pu}du + \int_1^\infty e^{-p\frac{u^3}{6}}du  \right] \nonumber \\
%\leq Ce^{-2p(1-K_1)\tau^3}\left[ \mathcal{O}\left(\frac{1}{\tau^2} \right)  + C\right] 
\leq Ce^{-2p(1-K_1)\tau^3}\\
\left\| e^{\Theta(\tau,b)-K(t)}\right\|_{L^\infty(R_3)} = e^{-2\tau(\sigma+\tau^2-\delta)}\leq Ce^{-2(1-K_1)\tau^3}
\end{gather}
where we set $s=2\sigma +2\tau^2 - \delta$, $0<\delta <  \sigma+ \tau^2$.
%provided $t<K_1(\sigma+\tau^2)$ with $0<K_1<1$. 
The proof for the other phases on the contours $L_3$, $R_2$ and $L_2$ is analogous.
\end{proof}

Before estimating the entries of the jump matrices on $J_2$ and $J_3$, we factor the jumps in the following way. We split the jump $J_2$ into two jumps (and two curves): with abuse of notation we call the first one $J_2$ and we merge the second jump with the jump on $R_1$. Thus, the new jumps are the following (see \figurename \ \ref{curvetau})
\begin{gather}
J_2 = \left[ \begin{array}{cccc}
1&0&e^{-\Theta(\tau,-a)+K(t)}&e^{-\Theta(\xi_-,s)} \\
0&1&e^{\Theta(\xi_+,t)}&0\\
0&0&1&0 \\
0&0&0&1 
\end{array} \right]
 \\
\tilde J_2 =    \left[ \begin{array}{cccc}
1&0&0&0 \\
e^{\Theta_{\wt \sigma}}&1&0 &-e^{\Theta(-\tau,b)+K(s)} \\
0&0&1&0 \\
0&0&0&1 
\end{array} \right]
\end{gather}

Analogously, we split the jump $J_3$ into two jumps: we call the first one again $J_3$ and we merge the second one with the jump on $L_1$. The new configuration of jump matrices is illustrated in \figurename \ \ref{curvetau}. 
\begin{gather}
J_3  = \left[ \begin{array}{cccc}
1&0&0&e^{-\Theta(\xi_-,s)} \\
0&1&e^{\Theta(\xi_+,t)}&e^{\Theta(-\tau,b)+K(s)} \\
0&0&1&0 \\
0&0&0&1 
\end{array} \right]
 \\
\tilde J_3 = \left[ \begin{array}{cccc}
1&e^{-\Theta_{\wt \sigma}}&e^{-\Theta(\tau,-a)+K(t)}&0\\
0&1&0&0 \\
0&0&1&0 \\
0&0&0&1 
\end{array} \right]
\end{gather}

\begin{lemma}\label{lemmaJ23}
Let $\kappa := \frac{8}{3} - \frac{p}{6}$. Given $0< K_2 <1$ fixed and $t = 4\tau^2 -\delta$, $0< \delta \leq K_2\kappa \tau^2$, then the $(1,3)$ and $(2,3)$ entries of the jump matrix $J_2$ are exponentially suppressed as $\tau \rightarrow + \infty$ in $L^p$ norms with $p=1,2, +\infty$. 

Given $0<K_3<1$ fixed and $s= \tau^2 + 2\sigma - \delta$, $0<\delta \leq K_3 \left(2\sigma + \frac{2}{3}\tau^2 \right)$, the $(2,4)$ entry of $\tilde J_2$ is exponentially suppressed in any $L^p$ norm ($1\leq p \leq \infty$).

Similarly, the same results hold true for the $(1,4)$ and $(2,4)$ entries of $J_3$ and for the $(1,3)$ entry of $\tilde J_3$.
\end{lemma}

\begin{proof}
The first row on $J_2$ is the same as the one on $L_2$ and the entry $e^{-\Theta(\tau,-a)+K(t)}$ is exponentially suppressed in any $L^p$ norm, thanks to Lemma \ref{lemmaLR23tau}. 

Regarding the remaining term on the second row, the real part of the argument in the exponent is
\begin{gather}
\Re\left[ \Theta(\xi_+,t) \right] = \frac{u^3}{6} - \frac{\tau}{2^{2/3}} {u}^{2} - \frac{\delta}{2 \sqrt[3]{2}}u - \frac{16}{3}{\tau}^{3}+2\tau \delta
\end{gather}
where we set $t = 4\tau^2 - \delta$, $ \delta>0$.
%Now, the two critical points where the monotone behaviour of $-\Theta(\tau,a)+K(s)$ changes are either negative or bigger than $\sqrt[3]{2}\tau$, when $J_2$ stops existing (at least for $\tau>>1$). In the interval $[0, \sqrt[3]{2}\tau]$ the function is negative in $0$ and always decreasing. 
\begin{oss}
A parametrization for the curve $J_2$ is $\lambda =  \sqrt[3]{2}\tau + u \left[ \frac{1}{2} \pm \frac{2}{\sqrt{3}}i  \right]$, $u\in [0, \sqrt[3]{2}\tau]$. When $u=\sqrt[3]{2}\tau$, the curve $J_2$ hits the curve $R_1$ and for $u > \sqrt[3]{2}\tau$ the contour $L_2$ appears. 
% the point of intersection with the curve $L_1$ (i.e. when $R_2$ becomes $J_2$) is at $u=\sqrt[3]{2}\tau$.
\end{oss}

Provided  $\delta < \kappa \tau^2$ ($\kappa:=\frac{8}{3}-\frac{p}{6}$), it is straightforward to compute the $L^p$ norms ($1 \leq p < 16$)
\begin{gather}
\left\| e^{ \Theta(\xi_+,t)}\right\|^p_{L^p(J_2)} =   2e^{-2p\tau\left( \frac{8}{3}\tau^2 -\delta\right)}\int_0^{\sqrt[3]{2}\tau} e^{p\left(\frac{u^3}{6} - \frac{\tau}{2^{2/3}} {u}^{2} - \frac{\delta}{2 \sqrt[3]{2}}u\right)}du \nonumber \\
\leq C e^{-2p\tau\left[ \left(\frac{8}{3}-\frac{p}{6}\right)\tau^2 -\delta\right]}\left[ \int_1^{\infty} e^{- \frac{p \tau}{2^{2/3}} {u}^{2}}du + \int_0^{1} e^{- \frac{p\delta}{2 \sqrt[3]{2}}u}du \right] \leq C e^{-2p \kappa (1-K_2)\tau^3} \\
\left\| e^{ \Theta(\xi_+,t)}\right\|_{L^\infty(J_2)} = e^{-\frac{16}{3}\tau^3 + 2\tau\delta} \leq C e^{- 2 \left[ \kappa (1-K_2) + 16\right] \tau^3}
\end{gather}
for some suitable $0<K_2<1$.

The phase on $\tilde J_2$ behaves like
\begin{gather}
\Re\left[ \Theta(-\tau,b) +K(s) \right] = - \frac{u^3}{6} - \frac{\tau}{2 \,2^{2/3}}{u}^{2}-\frac{\delta}{2 \, 2^{1/3}} u - \frac{2}{3}{\tau}^{3}-2\tau \sigma+\tau\delta
\end{gather}
where we set $s= \tau^2 +2\sigma-\delta$, $\delta>0$. Thus, provided $\delta < 2\sigma + \frac{2}{3}\tau^2$, the $L^p$ norms are
%and we get $\delta < 2\sigma + \frac{2}{3}\tau^2$.
\begin{gather}
\left\| e^{ \Theta(-\tau,b)+K(s)}\right\|^p_{L^p(\tilde J_2)} =   2 e^{-p\tau\left( \frac{2}{3}\tau^2 +2\sigma -\delta\right)}\int_0^{\sqrt[3]{2}\tau} e^{-p\left( \frac{u^3}{6} + \frac{\tau}{2 \,2^{2/3}}{u}^{2} + \frac{\delta}{2 \, 2^{1/3}} u \right)}du \nonumber \\
\leq C e^{-p\tau\left( \frac{2}{3}\tau^2 +2\sigma -\delta\right)} \left[ \int_0^{1} e^{-\frac{p\delta}{2 \, 2^{1/3}} u}du + \int_1^{\infty} e^{-p\frac{u^3}{6}}du\right] \nonumber \\
\leq C e^{-p\left( 1-K_3\right)\tau^3}  \\
\left\| e^{\Theta(-\tau,b)+K(s)}\right\|_{L^\infty(\tilde J_2)} = e^{-\frac{2}{3}\tau^3 -2\tau\sigma + \tau\delta} \leq e^{-C(1-K_3)\tau^3}
\end{gather}
for some suitable $0<K_3<1$.

The arguments for $J_3$ and $\tilde J_3$ are analogous. 
\end{proof}

\subsection{Global parametrix. The model problem}
We will now define a new ``model" RH problem which will eventually approximate the solution to our original problem $\hat \Gamma$.

We define the following RH problem:
\begin{equation}
\left\{ \begin{array}{ll}
\Omega_+(\lambda) = \Omega_-(\lambda) J_{\Ai}(\lambda) & \text{on} \ L_1\cup R_1 \\
\Omega_+(\lambda) = \Omega_-(\lambda) J_{R}(\lambda) & \text{on} \ L_2\cup R_2 \\
\Omega_+(\lambda) = \Omega_-(\lambda) J_{L}(\lambda) & \text{on} \ L_3\cup R_3 \\
\Omega(\lambda) = I + \mathcal{O} \left(\lambda^{-1}\right) & \text{as} \ \lambda \rightarrow \infty
\end{array}
\right.
\end{equation}
with jumps
\begin{gather}
J_{\Ai} :=  \left[ \begin{array}{cccc}
1&e^{-\Theta_{\wt \sigma}}\chi_{_{L_1}}&0 & 0\\
e^{\Theta_{\wt \sigma}}\chi_{_{R_1}}&1&0&0 \\
0 &0&1&0 \\
0&0&0&1 
\end{array} \right]  \\
J_R:= \left[ \begin{array}{cccc}
 1&0& 0& e^{- \Theta(\xi_-,s)}\chi_{_{L_2}} \\
0&1&0&0 \\
0&0&1&0 \\
e^{ \Theta(\xi_-,s)} \chi_{_{R_2}} &0&0& 1 
\end{array} \right] \\
J_L:= \left[ \begin{array}{cccc}
1&0&0&0 \\
0&1& e^{\Theta(\xi_+,t)}\chi_{_{R_3}}& 0 \\
0& e^{- \Theta(\xi_+,t)}\chi_{_{L_3}}& 1&0 \\
0&0&0&1 
\end{array} \right] 
\end{gather}

Let's denote by $\Psi_{a,b}$ the $4\times 4$ solution to the Airy RHP related to the submatrix formed by the $a$-th row and column and by the $b$-th row and column. In particular, we call $\Psi_{1,2}$ the matrix solution to the Hasting-McLeod Airy RHP for the minor $(1,2)$, related to the jump $J_\Ai$, with asymptotic solution
\begin{equation}
\Psi_{1,2}(\wt \sigma) =I_{4\times 4}  + \frac{1}{\lambda}\left[
 \begin{array}{cc|cc} 
 -p(\wt \sigma)& -q(\wt \sigma) & 0 & 0 \\
q(\wt \sigma) & p(\wt \sigma) & 0 & 0 \\ \hline
  0 &0&0&0 \\
  0&0&0& 0
 \end{array}  \right] + \mathcal{O}\left( \frac{1}{\lambda^2} \right)
\end{equation}
We consider now the matrix $\Xi:= \Omega \cdot \Psi_{1,2}^{-1}(\wt \sigma)$. This matrix doesn't have jumps on $L_1$ and $R_1$ by construction, but still has jumps on $L_2$, $R_2$ and $L_3$, $R_3$: 
%Basically we eliminated the extra jumps $L_1/R_1$ and we got back to the situation as when $\sigma \rightarrow +\infty$. 
%\textbf{CAVEAT!} The jumps that $\Gamma_1$ satisfies are actually different from $J_L$ and $J_R$. These jumps are conjugated by $\Psi_{1,2}$:
\begin{gather}
\tilde J_L := \Psi_{1,2}J_L \Psi_{1,2}^{-1} \ \ \ \text{and} \ \ \ \tilde J_R := \Psi_{1,2}J_R \Psi_{1,2}^{-1}
\end{gather}
On the other hand, as $\tau \rightarrow + \infty$ the critical points $\pm \sqrt[3]{2}\tau$ as well as the curves $L_2$, $R_2$, $L_3$, $R_3$ go to infinity, while the matrix $\Psi_{1,2}$ is asymptotically equal to the identity matrix.

We are left with 
\begin{gather}
\Xi =  \mathcal{E}_{1} \cdot \Psi_{2,3}(t) \cdot \Psi_{1,4}(s)
\end{gather}
where $\Psi_{2,3}$ and $\Psi_{1,4}$ where defined in (\ref{HMLsolution}) and $\mathcal{E}_1$ is the error matrix.

Following the previous remark, it is easy to show that the error matrix $\mathcal{E}_1$ is a sufficiently small perturbation of the identity and therefore we can apply the Small Norm Theorem and approximate the global parametrix $\Omega$ by simply the product of the matrices $\Psi_{a,b}$ ($(a,b) = (1,2), (2,3), (1,4)$)
%\begin{note}
%$\Psi_{2,3}$ and $\Psi_{1,4}$ commute. 
%\end{note}
%In conclusion, our global parametrix is
\begin{equation}
\Omega =\Xi \cdot \Psi_{1,2}(\wt \s) \sim \Psi_{2,3}(t) \cdot \Psi_{1,4}(s) \cdot \Psi_{1,2}(\wt \sigma).
\end{equation}

\subsection{Approximation and error term for the matrix $\hat \Gamma$}
The relation between our original RH problem $\hat \Gamma$ and the global parametrix $\Omega$ is the following
\begin{equation}
\hat \Gamma = \mathcal{E}_2 \cdot \Omega := \mathcal{E}_2 \cdot \Psi_{2,3}(t) \cdot \Psi_{1,4}(s) \cdot \Psi_{1,2}(\wt \sigma)
\end{equation}
where $\mathcal{E}_2$ is again an error matrix, to which we will apply Small Norm Theorem once again.

\begin{lemma}\label{lemmaerrortau}
%Given $0<K_1, K_2, K_3 <1$ fixed and $s,t < K_1(\sigma + \tau^2)$, $s = 4\tau^2 -\delta$ with $0< \delta \leq K_2(8/3-p/6) \tau^2$,  $t= \tau^2 + 2\sigma - \delta$ with $0<\delta \leq K_3 \left(2\sigma + \frac{2}{3}\tau^2 \right)$
In the estimates on $s,t$ stated in Lemmas \ref{lemmaLR23tau} and \ref{lemmaJ23}, the error matrix $\mathcal{E} = \hat \Gamma(\lambda) \Omega^{-1}(\lambda)$ solves a RH problem with jumps on the contours as indicated in \figurename \ \ref{curvetau} and of the following orders
\begin{gather}
\left\{ 
\begin{array}{ll}
\mathcal{E}_+(\lambda) = \mathcal{E}_-(\lambda) J_\mathcal{E}(\lambda) & \text{on} \ \Sigma \\
\mathcal{E}(\lambda) = I + \mathcal{O}\left(\lambda^{-1}\right) & \text{as} \ \lambda \rightarrow \infty
\end{array} \right. \\
J_\mathcal{E} =  \\
\left[ \begin{array}{cccc}
1 & 0  & \mathcal{O} (\tau^{-\infty})\chi_{_{L_2}} + \mathcal{O}(\tau^{-\infty})\chi_{_{\tilde J_3}} & \mathcal{O}(\tau^{-\infty})\chi_{_{J_3}} \\
0 &1&  \mathcal{O}(\tau^{-\infty})\chi_{_{J_2}} &\mathcal{O}(\tau^{-\infty})\chi_{_{R_3}} + \mathcal{O}(\tau^{-\infty})\chi_{_{\tilde J_2}} \\
 \mathcal{O}(\tau^{-\infty}) \chi_{_{R_2}} & 0 &1&0 \\
 0 & \mathcal{O}(\tau^{-\infty})\chi_{_{L_3}}  &0& 1 
\end{array} \right]  \nonumber
\end{gather}
where $\Sigma$ is the collection of all contours and the $\mathcal{O}$-symbols are valid for $L^1$, $L^2$ and $L^\infty$ norms.
\end{lemma}

\begin{proof}
Due to Lemmas \ref{lemmaLR23tau} and \ref{lemmaJ23}, we know from the estimates above that all the extra phases that appear in the original RH problem for $\hat \Gamma$ are bounded by an expontential function of the form $C_1e^{-C_2\tau^3}$. 
The jumps of the error problem are the remaining jumps appearing in the $\hat \Gamma$-problem conjugated with the global parametrix $\Omega$:
\begin{gather}
J_{\mathcal{E}} = \Omega^{-1} \left( I + \mathcal{O}(\tau^{-\infty}) \right) \Omega = I  +\mathcal{O}(\tau^{-\infty})
\end{gather}
The last equality follows from the fact that the solution $\Omega$ depends on $\tau$ with a growth that is smaller than the bound $C_1e^{-C_2\tau^3}$ that we have for the phases.
\end{proof}

Thus, the Small Norm Theorem can be applied
\begin{equation}
\left\| \mathcal{E} (\lambda) - I\right\| \leq \frac{C}{\text{dist}(\lambda, \Sigma)} \left( \left\|J_{\mathcal{E}} -I \right\|_1 + \frac{\left\|J_{\mathcal{E}} -I \right\|^2_2}{1-\left\|J_{\mathcal{E}} -I \right\|_\infty} \right) \leq \frac{C}{\text{dist}(\lambda, \Sigma)} e^{-K\tau}
\end{equation}
where $\Sigma$ is the collection of all contours, for some positive constants $C$ and $K$. The error matrix $\mathcal{E}$ is then found as the solution to an integral equation
%\begin{equation}
%\mathcal{E}(\lambda) = I + \int_\Sigma \frac{\mathcal{E}_-(w) \left(J_\mathcal{E}(\lambda) - I \right) \, dw}{2\pi i (w-\lambda)}
%\end{equation}
%and can be obtained by iterations
%\begin{equation}
%\mathcal{E}^{(0)}(\lambda) = I, \ \ \ \mathcal{E}^{(k+1)}(\lambda) =  I + \int_\Sigma \frac{\mathcal{E}^{(k)}_-(w) \left(J_\mathcal{E}(\lambda) - I \right) \, dw}{2\pi i (w-\lambda)}
%\end{equation}
and, thanks to Lemma \ref{lemmaerrortau} we have
\begin{equation}
\mathcal{E}(\lambda) = I + \frac{1}{\text{dist} (\lambda, \Sigma)}\mathcal{O} \left(\tau ^{-\infty} \right).
\end{equation}

We need the first coefficient $\hat \Gamma_1 = \hat \Gamma_1(s,t,\wt \sigma)$ of $\hat \Gamma (\lambda)$ at $\lambda = \infty$ and how it compares to the corresponding coefficient $\Omega_1$ of $\Omega(\lambda)$; the error analysis above shows that
\begin{equation}
\hat \Gamma_1 = \Omega_1 + \mathcal{O} \left( \tau^{-\infty}\right). \label{GammaOmega1}
\end{equation}

\subsection{Conclusion of the proof of Theorem \ref{THMtau}}

\begin{thm}
The Fredholm determinant $\det (\operatorname{Id} - \hat \Pi \mathbb{H} \hat \Pi)$ is  equal to the Jimbo-Miwa-Ueno isomonodromic $\tau$-function of the RH problem (\ref{RHH}). For any parameter $\rho$ on which the integral operator $\hat\Pi \mathbb{H} \hat \Pi$ may depend, we have
\begin{gather}
\partial_{\rho} \ln \det (\operatorname{Id} - \hat \Pi \mathbb{H} \hat \Pi) = \omega_{JMU} (\partial_\rho) = \int_\Sigma \operatorname{Tr} \left( \Gamma_-^{-1}(\lambda) \Gamma'_-(\lambda) \partial_\rho \Xi(\lambda)\right) \frac{d\lambda}{2\pi i} \label{JMUtauH}
%\partial_{t} \ln \det (\operatorname{Id} - \hat \Pi \mathbb{H} \hat \Pi) = \omega_{JMU} (\partial_t) = \int_\Sigma \operatorname{Tr} \left( \Gamma_-^{-1}(\xi) \Gamma'_-(\xi) \partial_t \Xi(\xi)\right) \frac{d\xi}{2\pi i}
\end{gather}
More specifically,
\begin{align}
\partial_{\wt \s} \ln \det (\operatorname{Id}- \hat \Pi \mathbb{H}\hat \Pi) &= -  \operatorname{res}_{\lambda = \infty} \operatorname{Tr} \left(\Gamma^{-1}\Gamma'\partial_{\wt \s} T \right) =\frac{1}{\lambda}\Gamma_{1; \, (2,2)} \\
\partial_t \ln \det (\operatorname{Id}- \hat \Pi \mathbb{H}\hat \Pi) &= -  \operatorname{res}_{\lambda = \infty} \operatorname{Tr} \left(\Gamma^{-1}\Gamma'\partial_t T \right) =-\frac{1}{\sqrt[3]{2}\lambda}\Gamma_{1; \, (3,3)} \\
\partial_s \ln \det (\operatorname{Id} - \hat \Pi\mathbb{H}\hat \Pi) &= -  \operatorname{res}_{\lambda = \infty} \operatorname{Tr} \left(\Gamma^{-1}\Gamma'\partial_s T \right)  = \frac{1}{\sqrt[3]{2}\lambda}\Gamma_{1; \, (4,4)} 
\end{align}
where $\Gamma_1 := \lim_{\lambda \rightarrow \infty} \lambda (\Gamma(\lambda) - I)$.
\end{thm}

\begin{proof}
The first part of the Theorem is the same as Theorem \ref{thmtaufunctionresiduessigma}.
Then, using the Jimbo-Miwa-Ueno residue formula, we have
\begin{equation}
\partial_{\rho} \ln \det (\operatorname{Id}  - \hat \Pi\mathbb{H}\hat \Pi) = -  \operatorname{res}_{\lambda = \infty} \text{Tr} \left(\Gamma^{-1}\Gamma'\partial_{\rho} T \right) 
\end{equation}
with $\rho = \wt \s, s, t$.

Taking into account the definition of the conjugation matrix $T$ (see (\ref{conjT})) and the asymptotic behaviour of the matrix $\Gamma$ at infinity, we get again
\begin{gather}
\operatorname{Tr} \left[ \Gamma^{-1}\Gamma'\partial_{\wt \sigma}T \right] = - \frac{\Gamma_{1; \, (2,2)}}{\lambda}, \ \  \
\operatorname{Tr} \left[ \Gamma^{-1}\Gamma'\partial_sT \right] = - \frac{\Gamma_{1; \, (4,4)}}{\sqrt[3]{2}\lambda}, \ \  \
\operatorname{Tr} \left[ \Gamma^{-1}\Gamma'\partial_tT \right]  =  \frac{\Gamma_{1; \, (3,3)}}{\sqrt[3]{2}\lambda}.
\end{gather}
\end{proof}

On the other hand, thanks to Lemma \ref{lemmaerrortau} and the Small Norm Theorem, we can approximate the solution $\Gamma$ with the global parametrix $\Omega$ using (\ref{GammaOmega1}) and we get
\begin{gather}
\text{d} \ln \det \le(\operatorname{Id}  - \mathbb{H}\bigg|_{[-\sigma-\tau^2+s, \sigma+\tau^2-t]}\ri) =
%= \partial_s \ln \tau \text{d}s + \partial_t \ln \tau \text{d}t + \partial_{\wt \sigma} \ln \tau \text{d}\wt \sigma 
\nonumber \\
 p(s) \text{d}s +p(t) \text{d}t + p(\wt \sigma) \text{d}\wt \sigma + \mathcal{O}\left( \tau^{-\infty}\right)  \text{d}s + \mathcal{O}\left( \tau^{-\infty}\right)  \text{d}t + \mathcal{O}\left( \tau^{-\infty}\right)  \text{d}\wt \sigma \nonumber \\
 + \mathcal{O}\left( \tau^{-\infty}\right)  \text{d}s \, \text{d}t + \mathcal{O}\left( \tau^{-\infty}\right)  \text{d}s \, \text{d}\wt \s + \mathcal{O}\left( \tau^{-\infty}\right)  \text{d}t \, \text{d}\wt \s + \mathcal{O}\left( \tau^{-\infty}\right)  \text{d}s \, \text{d}t  \, \text{d}\wt \s.
 \end{gather}
Integrating from a fixed point $(s_0, t_0, \wt \sigma_0)$ 
%(where $s_0:= a+ \sigma + \tau^2$ and $t_0 = -b+\sigma + \tau^2$) 
up to $(s,t,\wt \sigma)$,
 \begin{gather}
 \det   \le(\operatorname{Id}  - \mathbb K^{\rm tac}\bigg|_{[-\sigma-\tau^2+t, \sigma+\tau^2-s]}\ri) =\nonumber \\ 
 \frac{e^C \det \le(\operatorname{Id}  - K_{\Ai}\bigg|_{[s,+\infty)}\ri) \det \le(\operatorname{Id}  - K_{\Ai}\bigg|_{[t,+\infty)}\ri)  \det \le(\operatorname{Id}-K_{\Ai}\bigg|_{[\wt \s,\infty)}\ri) \left(1+ \mathcal{O}(\tau^{-1}) \right)}{ \det \le(\operatorname{Id}-K_{\Ai}\bigg|_{[\wt \s,\infty)}\ri)} \nonumber \\
 = e^C \det \le(\operatorname{Id}  - K_{\Ai}\bigg|_{[s,+\infty)}\ri) \det \le(\operatorname{Id}  - K_{\Ai}\bigg|_{[t,+\infty)}\ri) \left(1+ \mathcal{O}(\tau^{-1}) \right) \label{Cint}
\end{gather}
with $s,t $ within the domain that guarantees the uniform validity of the estimates above (see Lemmas \ref{lemmaLR23tau} and \ref{lemmaJ23}) and $C =  \ln \det (\operatorname{Id}  - \mathbb{H}\chi_{[-\sigma_0-\tau^2+t_0, \sigma_0+\tau^2-s_0]})$.

Finally, we need again to show that the constant of integration $C$ is equal zero. 
\begin{lemma}
The constant of integration $C$ in the formula (\ref{Cint}) is zero.
\end{lemma}

\begin{proof}
First of all we notice that (see Lemma \ref{conditionedprob})
\begin{gather}
\det (\operatorname{Id} - \Pi \mathbb K^{\rm tac} \Pi) = \frac{\det \left(\operatorname{Id} - \hat \Pi \mathbb{H}\hat \Pi \right)}{\det \left( \operatorname{Id} -\pi  K_{\Ai}\pi \right)}  \nonumber \\
= \det \left( \left[ \begin{array}{c|c} (\text{Id}- \pi K_\Ai \pi)^{-1}  & 0 \\ \hline  0 & \text{Id}\end{array} \right] \cdot \left[\begin{array}{c|c} \text{Id}- \pi K_\Ai \pi & \sqrt[6]{2} \pi \mathfrak{A}_{-\tau}^T\tilde \Pi \\ \hline\sqrt[6]{2}\tilde \Pi \mathfrak{A}_{\tau} \pi & \text{Id} - \sqrt[3]{2}\tilde \Pi K_\Ai^{(\tau,-\tau)} \tilde \Pi  \end{array} \right] \right) \nonumber \\
= \det \left( \text{Id} - \left[\begin{array}{c|c} 0 & - \sqrt[6]{2}(\text{Id}- \pi K_\Ai \pi)^{-1}   \pi \mathfrak{A}_{-\tau}^T \tilde \Pi \\ \hline -\sqrt[6]{2}\tilde \Pi \mathfrak{A}_{\tau} \pi &\sqrt[3]{2} \tilde \Pi K_\Ai^{(\tau,-\tau)} \tilde \Pi  \end{array} \right] \right) 
\end{gather}
where $\hat \Pi := \pi \oplus \tilde \Pi$, $\pi$ is the projector on $[\wt \sigma, \infty)$ and $\Pi$ is the projector on $[\tilde a,\tilde b]$. 

Along the same guidelines as the proof of Lemma \ref{Cequal0}, we will perform some uniform estimates on the entries of the kernel that will lead to the desired result. 

We have
\begin{align}
\left|\sqrt[3]{2}\tilde \Pi K_\Ai^{(\tau,-\tau)}(u,v)\tilde \Pi \right| & \leq \frac{C}{ \tau} e^{-\frac{4}{3}\tau^3 -2\tau \sigma +2\tau v }  \leq \frac{C_1}{\sqrt{\tau}}\\
 \left|\sqrt[6]{2}\tilde \Pi  \mathfrak{A}_\tau (x,v) \pi \right| & \leq \frac{C_2}{\sqrt{\tau}} \\
\left| \sqrt[6]{2} (\text{Id}- \pi K_\Ai \pi)^{-1}  \pi \mathfrak{A}^T_{-\tau}(u,y) \tilde \Pi \right| & \leq C_\Ai e^{-\frac{4}{3}\tau^3} \left[ e^{- 2\tau( y - \sigma + \sqrt[3]{2}u) }  + \frac{ e^{2\tau\left(y-  \sigma\right) }}{2\sqrt[3]{2}\tau}\right]   \nonumber \\
 & = C_\Ai e^{-\frac{\tau^3}{3}}  \left[ e^{ -\tau^3 - 2\tau( y - \sigma + \sqrt[3]{2}u) }   + \frac{ e^{-\tau^3 +2\tau\left(y-  \sigma\right) }}{2\sqrt[3]{2}\tau}\right]  \leq \frac{C_3}{\sqrt{\tau}} e^{-\tau(u-\sigma)}
\end{align}
for some positive constants $C_j$ ($j=1, 2, 3$), where the variables $x,y$ run in $[\tilde a,\tilde b]$ and $u,v$ run in $[\wt \sigma, \infty)$. Such estimates follow again from simple arguments on the asymptotic behaviour of the Airy function when its argument is very large. Moreover, the resolvent of the Tracy-Widom distribution is uniformly bounded and independent on $\tau$; here is the reason for the constant $C_\Ai$.

Collecting the above estimates, we have
\begin{gather}
\left[ \begin{array}{c|c} 0 & - \sqrt[6]{2} (\text{Id} - \pi K_\Ai \pi )^{-1} \pi \mathfrak{A}^T_{-\tau}\tilde \Pi \\  \hline - \sqrt[6]{2} \tilde \Pi \mathfrak{A}_{\tau} \pi  &\sqrt[3]{2}\tilde \Pi K_\Ai^{(\tau,-\tau)}\tilde \Pi  \end{array}\right] \leq C_\tau \left[ \begin{array}{c|c} 0 & f(u) \\  \hline 1  & 1  \end{array}\right] 
%\leq C_\tau \left[ \begin{array}{c|c} f(u) & f(u) \\  \hline 1  & 1  \end{array}\right]
\end{gather}
with $ C_{\tau} := \frac{\max \{C_j, \ j=1,2,3  \}}{\sqrt{\tau}}$ and $f(u) = e^{-\tau (u-\sigma)} $. 
On the right hand side, we have a new operator $\mathcal{M}$ acting on $L^2([\wt \s, \infty)) \oplus L^2([\tilde a,\tilde b]))$ with bounded trace
\begin{equation}
\operatorname{Tr} \mathcal{M} \leq \hat C \left( \left\| f \right\|_{L^2(\wt \sigma, \infty)}^2 + (\tilde b-\tilde a)\right) \leq \hat C(b-a)
\end{equation}
for some positive constant $\hat C$, since $\left\| f \right\|_{L^2(\wt \sigma, \infty)}^2 \rightarrow 0$ as $\tau \rightarrow + \infty$.

Concluding, having $[a,b]$ fixed,
\begin{gather}
|\ln \det (\operatorname{Id}  - \Pi \mathbb K^{\rm tac} \Pi)| =  \sum_{n=1}^\infty \frac{\text{Tr} \, (\Pi \mathbb K^{\rm tac} \Pi)^n}{n} \leq  \sum_{n=1}^\infty \frac{C_\tau^n \hat C^n (b-a)^n}{n} \nonumber \\
\leq \sum_{n=1}^\infty C_\tau^n \hat C^n (b-a)^n = \frac{C_\tau \hat C (b-a)}{1-C_\tau \hat C (b-a)} \rightarrow 0
\end{gather}
as $\tau \rightarrow + \infty$.

Therefore, the constant of integration is equal zero.
\end{proof} 

\section*{Acknowledgements}
The author would like to acknowledge Dr. Marco Bertola and Dr. Mattia Cafasso for proposing this problem and for their help in setting it.

%\nocite{AdlerFerrariVMoer}
%\nocite{AdlerJohanssonVMoer}
%\nocite{Delvaux}
%\nocite{DelvauxKuijZhang}
%\nocite{FerrariVeto}
%\nocite{glorystealer}
%\nocite{Johansson}
%\nocite{JohanssonAiry}
%\nocite{JohanssonTac}
%\nocite{MeM}
%\nocite{MeMmulti}
%\nocite{MeMnum}
%\nocite{Spohn}
%\nocite{TWAiry}

%\bibstyle{plain}
\bibliography{Tacnode2}
%\bibliography{Tacbiblio}

\end{document}